\documentclass{article}

\usepackage{amssymb}
\usepackage{amsthm}
\usepackage{amsmath}
\usepackage{thmtools,thm-restate}
\usepackage{array,amsmath,calc}
\usepackage{mathtools}
\usepackage{geometry}[margin=.5in]
\usepackage{enumitem}
\usepackage{soul}

\usepackage{hyperref}

\newtheorem{theorem}{Theorem}[section]
\newtheorem{lemma}[theorem]{Lemma}
\newtheorem{corollary}[theorem]{Corollary}
\newtheorem{proposition}[theorem]{Proposition}
\newtheorem{definition}[theorem]{Definition}
\newtheorem{remark}[theorem]{Remark}

\usepackage{xcolor}

\usepackage{algorithm}
\usepackage{algpseudocode}
\newenvironment{varalgorithm}[1]
	{\algorithm\renewcommand{\thealgorithm}{#1}}
	{\endalgorithm}

\usepackage{tikz,tikz-3dplot}
\usetikzlibrary{shapes.geometric}
\tikzstyle{scorestars}=[star, star points=5, star point ratio=2.25, draw,inner sep=1pt,anchor=center]
\usetikzlibrary{calc}
\usetikzlibrary{shapes}
\usetikzlibrary{plotmarks}
\usepackage{float}

\usepackage{caption}
\usepackage{subcaption}

\makeatletter
\newcommand*{\textoverline}[1]{$\overline{\hbox{#1}}\m@th$}
\makeatother

\begin{document}

\title{$O\left(1/T\right)$ Time-Average Convergence in a Generalization of Multiagent Zero-Sum Games} 
\author{James P. Bailey}

\date{}

\maketitle
\begin{abstract}
We introduce a generalization of zero-sum network multiagent matrix games and prove that alternating gradient descent converges to the set of Nash equilibria at rate $O(1/T)$ for this set of games. 
Alternating gradient descent obtains this convergence guarantee while using fixed learning rates that are four times larger than the optimistic variant of gradient descent.
Experimentally, we show with 97.5\% confidence that, on average, these larger learning rates result in time-averaged strategies that are 2.585 times closer to the set of Nash equilibria than optimistic gradient descent.
\end{abstract}

\section{Introduction}
We study online learning algorithms applied to network matrix games in the form
\begin{align*}
\max_{x_i\in \mathbb{R}^{S_i}} \left\langle x_i, \sum_{i\neq j} A^{(ij)} x_j - b_i\right\rangle \forall \ i=1,...,N. 
\end{align*} 
These games are used to capture a network where an agent receives utility based on their interactions with other agents, e.g., agent $i$ receives utility $\langle x_i, A^{(ij)}x_j\rangle$ when agent $i$ selects action $x_i$ and agent $j$ selects action $x_j$. 
A solution to this game is known as a Nash equilibrium, $x^*$, and is given by
\begin{align*}
\left\langle x^*_i, \sum_{i\neq j} A^{(ij)} x^*_j\right\rangle \geq \left\langle x_i, \sum_{i\neq j} A^{(ij)} x^*_j\right\rangle \forall x_i\in \mathbb{R}^{S_i} \forall \ i=1,...,N, 
\end{align*}
i.e., no agent can obtain a better outcome by deviating from $x^*$. 

Zero-sum network games, equivalently zero-sum polymatrix games \cite{cai2016zero}, are a special case where $A^{(ij)}= -[A^{(ji)}]^\intercal$ for all pairs of agents -- equivalently, $\langle x_i,  A^{(ij)} x_j \rangle +\langle x_j,  A^{(ji)} x_i \rangle=0$. 
Online learning dynamics and algorithms in zero-sum games have received a great deal of attention due to their numerous applications in areas such as Generative Adversarial Networks (GANS) \cite{goodfellow2014generative}, bargaining and resource allocation problems \cite{Shahrampour20OnlineAllocation}, and policy evaluation methods \cite{du2017stochastic}.

In each of these settings, the goal is to find a Nash equilibrium via online optimization techniques by having agents repeatedly play the game while updating their actions using only information about cumulative payouts, i.e., agent $i$ has access to only $\{\sum_{j\neq i} A^{(ij)}x_j^t\}_{t=0}^{T-1}$ when selecting strategy $x_i^T$ where $x_j^t$ is agent $j$'s action in the $t$-th game. 
While there are methods that guarantee last-iterate convergence  (e.g., \cite{daskalakis2019last, wei2020linear, abernethy2021last}), most methods focus on time-average convergence ($\sum_{t=0}^{T-1} x^t/T \to x^*$) since these methods tend to be faster (see e.g., \cite{golowich2020last}). 

The standard strategy for establishing time-average convergence relies on a connection between convergence and regret, a standard measure of performance in online optimization. 
Specifically, agent $i$'s regret for not playing $x_i$ is the difference between $i$'s cumulative utility and the cumulative utility had $i$ played $x_i$ instead.  
Formally $regret(x_i)= \sum_{t=0}^{T-1} \langle x_i-x_i^t, \sum_{j\neq i} A^{(ij)}x_j^t\rangle$. 
It is well-known that $f(T)$ time-average regret for all agents implies $f(T)$ time-average convergence to the set of Nash equilibrium in bounded zero-sum games (see \cite{cesa2006prediction}).

While there are several algorithms that obtain $O(1/T)$ time-average regret and convergence for zero-sum games \cite{kangarshahi2018let, mokhtari2020convergence}, no such results are known for general-sum games (no restrictions on $A^{(ij)}$). 
Recently, $poly(\log (T)/T)$ time-average regret has been shown in general-sum games \cite{daskalakis2021nearoptimal}.
However, this is insufficient for quickly finding Nash equilibria; $f(T)$ time-average regret in these settings only implies $f(T)$ time-average convergence to the set of coarse correlated equilibria -- a significantly weaker solution concept.  

To provide finer distinctions between the types of games, \cite{kannan2010games} introduces a hierarchy to capture all games.  
In the two agent settings, the rank of game is denoted by $rank(A^{(ij)}+ [A^{(ji)}]^\intercal)$ implying a two-agent game is zero-sum if and only if it is rank-0. 
Standard algorithms for finding Nash equilibria in zero-sum games are known to not work well even in rank-1 games \cite{balcan2012weighted} and other fast methods to find Nash equilibria for rank-1 games have been developed \cite{adsul2021fast}. 
In this paper, we focus on fast time-average convergence for a different generalization of zero-sum games. 

\subsection{Our Motivations}

Our methodology is heavily motivated by continuous-time optimization in games where agents' strategies are a continuous function of other agents' actions (see e.g., \cite{mertikopoulos2016learning}). 
In particular, continuous-time variants of follow-the-regularized-learner algorithms (FTRL), e.g., gradient descent and multiplicative weights, are known to achieve $O(1/T)$ time-average regret in general-sum games \cite{Mertikopoulos2018CyclesAdverserial}.
In the setting of zero-sum games, these learning \emph{dynamics} maintain constant energy and cycle around the set of Nash equilibria \cite{Mertikopoulos2018CyclesAdverserial} on closed orbits.

However, this is drastically different than what we see from discrete-time FTRL where agent strategies diverge from the set of Nash equilibria \cite{Bailey18Divergence}.  
This is because these algorithms are poor approximations of the continuous-time dynamics. 
Continuous-time variants of FTRL have been shown to form a Hamiltonian dynamic \cite{Bailey19Hamiltonian}, a well-known concept used to capture the evolution of a physical system.  
Discrete-time FTRL can be formulated by applying Euler integration to this Hamiltonian system; 
regrettably Euler integration is well-known to be a poor approximator of Hamiltonian systems. 
Instead, we focus on symplectic integrators (see e.g., \cite{Hairer2006EnergyConserve,hairer2006geometric}), which were designed for Hamiltonian systems. 
Specifically, we study alternating gradient descent, which arises naturally by applying Verlet integration, a symplectic technique, to continuous-time gradient descent.

\subsection{Our Contributions}

We prove that multi-agent alternating gradient descent achieves $O(1/T)$ time-average convergence to the set of Nash equilibrium in network zero-sum games (Theorem \ref{thm:multiResult}) matching the best known bound for convergence in zero-sum games. 
We show that alternating gradient descent accomplishes the convergence guarantee with learning rates up to four times larger than optimistic gradient descent.
Our theoretical work suggests that these larger learning rates translate to faster optimization guarantees (Theorems \ref{thm:large1} and \ref{thm:large2}). 
Our experiments support this; experimentally we show with 97.5\% confidence that, on average, alternating gradient descent results in time-averaged strategies that are 2.585 times closer to the set of Nash equilibria than optimistic gradient descent. 

Moreover, we introduce a generalization of the zero-sum network games, and show alternating gradient descent also achieves $O(1/T)$ time-average convergence to the set of Nash equilibria. 
In this generalization, we allow each agent to multiply their payoff matrices by an arbitrary positive-definite matrix.  Formally, a network positive negative definite game is given by 
\begin{align*}
&\max_{x_i\in \mathbb{R}^{S_i}} \left\langle x_i, P_i\sum_{i\neq j} A^{(ij)} x_j-b_i\right\rangle \forall \ i=1,...,N \\
&\text{where \ $P_i$ \ is \ positive-definite} \\
&\text{and} \ A^{(ij)}=-[A^{(ji)}]^\intercal \ \forall \{i,j\}\in [N]
\end{align*}
Our proposed methods allow us to extend important convergence results to settings that are adversarial in nature, but not necessarily zero-sum. 
We remark that our generalization is distinct from the rank-based hierarchy of bimatrix games introduced by \cite{kannan2010games}.
Specifically, the set of positive-negative definite games includes games at every level of the hierarchy. 
Further, unlike zero-sum games, an agent's payoff reveals no information about the payoff of other agents -- even in the 2-agent case. 

We accomplish this by showing that alternating gradient descent behaves similarly to its continuous-time analogue.  Specifically it has (i) an invariant energy function capturing all updates (Theorem \ref{thm:MultiEnergy2}), (ii) these energy functions are bounded (Theorem \ref{thm:BoundedOrbits}) and (iii) strategies approximately cycle (Theorem \ref{thm:Recurrence}). 
Finally, we relate the time-average of the strategies directly to the cyclic nature of the algorithm to prove $O(1/T)$ time-average convergence. 

In addition, we also prove several important properties of alternating gradient descent in general-sum games. 
Most notably, an agent using alternating gradient descent has $O(1/T)$ regret after agent 1 updates regardless of the opponents' strategies (Theorem \ref{thm:MultiRegret}). 
We remark that alternating gradient descent is unique relative to other learning algorithms in that agents take turns updating; as such, agent 1's regret is not necessarily $O(1/T)$ after other agents update and therefore Theorem \ref{thm:MultiRegret} cannot be directly compared to regret guarantees for other algorithms, e.g., \cite{daskalakis2021nearoptimal} remains the best guarantee for the standard notion of regret in general-sum games.

\section{Preliminaries}

We study repeated matrix network games between $N$ agents where each agent receives  utility based on their interactions with other individual agents. 
Agent $i$'s set of available actions are given by a convex space ${\cal X}_i$.  
For most of this paper, we use ${\cal X}_i=\mathbb{R}^{S_i}$  for some positive integer $S_i$. 
Once selecting strategies, $x=(x_1,...,x_N)\in {\bigtimes_{i=1}^n{\cal X}_i}$, agent $i$ receives a utility of $\langle x_i, A^{(ij)}x_j\rangle$ for the interaction between agents $i$ and $j$ where $i\neq j$. 
This yields the following network game where each agent seeks to maximize their individual utilities. 
\begin{align*}
	\max_{x_i\in {\cal X}_i} \left\langle x_i, \sum_{i\neq j} A^{(ij)} x_j\right\rangle \ for \ all \ i=1,...,N \tag{Network Matrix Game}
\end{align*}
The term $A^{(ij)}$ denotes agent $i$'s \emph{payoff} matrix against agent $j$. 
A solution to this game is known as a Nash equilibrium, $x^*$, and is characterized by
\begin{align*}
\left\langle x^*_i, \sum_{i\neq j} A^{(ij)} x^*_j\right\rangle \geq \left\langle x_i, \sum_{i\neq j} A^{(ij)} x^*_j\right\rangle \forall x_i\in \mathbb{R}^{S_i} \forall \ i=1,...,N, \tag{A Nash Equilibrium}
\end{align*}
i.e., no agent can obtain a better outcome by deviating from $x^*$. 
When ${\cal X}_i$ is affine and full-dimensional, an equivalent condition for a Nash equilibrium is given by $\sum_{j\neq i} A^{(ij)}x^*_j= \vec{0}$ since otherwise agent $i$ could move their strategy in the direction $\sum_{j\neq i} A^{(ij)}x^*_j$ to increase their utility.
Therefore $x^*$ is a Nash equilibrium if and only if $\sum_{j\neq i} A^{(ij)}x^*_j= \vec{0}$ for each agent $i$. 
When ${\cal X}_i=\mathbb{R}^{S_i}$, as is this case in most of this paper, $x^*_i=\vec{0}$ always corresponds to a Nash equilibrium. 
However in Section \ref{sec:bilinear} we extend our results to the utility function $\langle x_i , \sum_{j\neq i} A^{(ij)} x_j -b_i\rangle$ where Nash equilibria can be arbitrarily located.

In addition to general-sum games (no restrictions on $A^{(ij)}$), we also consider two other standard types of games -- zero-sum and coordination games. 
\begin{definition}
	A network game is a zero-sum network game iff $A^{(ij)}= -\left[ A^{(ji)}\right]^\intercal$ for all $i\neq j$. 
\end{definition}
\begin{definition}
	A network game is a coordination network game iff $A^{(ij)}= \left[ A^{(ji)}\right]^\intercal$ for all $i\neq j$. 
\end{definition}

In a zero-sum network game, agent $j$ loses whatever agent $i$ gains from their interaction. 
By \cite{cai2016zero} every zero-sum polymatrix game (a multiagent game where payouts are determined by tensors) is a zero-sum game and we lose no generality by replacing every instance of ``zero-sum game'' with ``zero-sum polymatrix game". 
At the other end of spectrum, agent $i$ and agent $j$ always have the same gains from their interactions in a coordination game. 
While our main results are for generalizations of zero-sum games, we also include several results for general-sum games and a generalization of coordination games.

\subsection{Online Optimization in Games}

Our primary interest is in repeated games. 
In this setting, each agent selects a sequence of strategies $\{x_i^0,...,x_i^T\}\subset {\cal X}_i$ and agent $i$ receives a cumulative utility of $\sum_{t=0}^T \langle  x_i, \sum_{j\neq i} A^{(ij)} x_j^t\rangle$. 
In most applications, $x_i^t$ is selected after seeing the gradient of the payout from the previous iteration, i.e., after seeing $\sum_{j\neq i} A^{(ij)} x_{j}^{t-1}$. 
Gradient descent (Algorithm \ref{alg:GradientMulti}) is one of the most classical algorithms for updating strategies in this setting.

\begin{varalgorithm}{SimGD}
	\caption{Gradient descent with simultaneous updates}\label{alg:GradientMulti}
	\begin{algorithmic}[1]
		\Procedure{SimGD}{$A,x^0,\eta$}\Comment{Payoff matrices, initial strategies and learning rates}
		\For{\texttt{$t=1,...,T$}}
			\For{\texttt{$i=1,...,N$}}
				\State $x_i^t:= x_i^{t-1} + {\eta_i} \sum_{j\neq i} A^{(ij)}x_j^{t-1}$ \Comment{Update strategies based on previous iteration}
			\EndFor
		\EndFor
		\EndProcedure
	\end{algorithmic}
\end{varalgorithm}

The learning rate $\eta_i>0$ describes how responsive agent $i$ is to the previous iterations. 
Typically in applications of Algorithm \ref{alg:GradientMulti}, $\eta_i$ decays  over time in order to prove $O(1/\sqrt{T})$ time-average regret and convergence when ${\cal X}$ is compact.
However, this decaying learning rate may not be necessary in general; \cite{Bailey19GDRegret} shows the same $O(1/\sqrt{T})$ guarantees in 2-agent, 2-strategy zero-sum games with an arbitrary fixed learning rate and provides experimental evidence to suggest the results extend to larger games.
In this paper, we consider variations of gradient descent in order to improve optimality and convergence guarantees. 
The variants we consider all rely on time-invariant learning rates that are independent of the time horizon $T$ and yield stronger optimization than the classical method of gradient descent with simultaneous updates. 

%\begin{algorithm}[H]
%	\SetAlgoLined
%	%	\KwResult{Write here the result }
%	\textbf{Input:} \\
%	\hspace*{1.2em} Agent $i's$ payoff matrix against agent $j$: $A^{(ij)}$ for $i=1,...,n$ and $j\neq i$\\
%	\hspace*{1.2em} Agent $i$ initial strategy: $x_i^0\in {\cal X}_i$ for $i=1,...,n$\\
%	\hspace*{1.2em} Agent $i$ learning rate: ${\eta}_i>0$ for $i=1,...,n$\\
%	
%	
%	\For{$t=1,...,T$}{
%		\For{$i=1,...,N$}{
%			$x_i^t:= x_i^{t-1} + {\eta_i} \sum_{j\neq i} A^{(ij)}x_j^{t-1}$ \tcp*{Agent Updates Strategies Based on the Outcome of the Previous Iteration}
%		}
%	}
%	\caption{Gradient Descent with Simultaneous Updates.}\label{alg:GradientMulti}
%\end{algorithm}

\section{Alternating Gradient Descent in 2-Agent Games}
\label{sec:2Agent}

We begin by closely examining a 2-agent game.  
For reasons which will become apparent later, we will simplify the notation so that $x\in{\cal X}$ describes agent $1$'s strategy space, $y\in{\cal Y}$ describes agent 2's strategy space, and $A^{(12)}=A$ and $B=A^{(21)}$ describe the agent's payoff matrices respectively.  
This results in the following game.

\begin{equation}\label{eqn:2AgentGame}
\tag{2-Agent Game}
\begin{aligned}
	\max_{x\in {\cal X}} \ \langle x, Ay\rangle \\
	\max_{y\in {\cal Y}} \ \langle y, Bx\rangle \\
\end{aligned}
\end{equation}

In this section, we analyze alternating gradient descent (Algorithm \ref{alg:2Agent} below) in 2-agent games and show four properties for general-sum games:
\begin{enumerate}
	\item \textbf{Regret:} An agent has $O\left({1}/{T}\right)$ time-average regret immediately after updating if they use alternating gradient descent with an arbitrary vector of fixed learning rates against an arbitrary opponent with an unknown time horizon $T$ (Theorem \ref{thm:regret}). 
	We remark that the $O\left({1}/{T}\right)$ guarantee does not hold after the opposing agent updates (Proposition \ref{prop:BadRegret}). \label{1}
	\item \textbf{Large Learning Rates Work Well:} Optimization guarantees of alternating gradient descent improve as we use larger learning rates.  Specifically, agent $1$ is guaranteed a utility of $-\langle x^0, D_\eta^{-1} x^0\rangle \to 0$ as $\eta\to \infty$ and, against an unresponsive opponent, agent 1's utility after updating goes to $\infty$ as $\eta\to \infty$ (Theorems \ref{thm:large1} and \ref{thm:large2}). \label{2}
	\item \textbf{Self-Actualization:} In order to maximize agent $1$'s regret  for not playing the fixed strategy $x$, agent two will actually force agent 1 to play the strategy $x$.  
	Formally, for any sequence $\{y^0,...,y^T\}$ that maximizes agent 1's regret for using $\{x^0,...,x^{T}\}$ from alternating gradient descent instead of the fixed strategy $x$ will result in $x^{T+1}=x$ (Theorem \ref{thm:Actualization}). \label{3}
%	\item \textbf{Smooth Energy Function for Discrete Updates:} If both agents use alternating gradient descent with arbitrary vectors of fixed learning rates then there is an invariant, continuous energy function that captures all updates (Theorem \ref{thm:energy}).\label{4}
	\item \textbf{Volume Preservation:} Alternating gradient descent preserves the volume of every measurable set of initial conditions when agents use arbitrary learning rates  (Theorem \ref{thm:Volume})\label{item:volume}.\label{5}
\end{enumerate}

We show and explore the meaning of each of these properties in Sections \ref{sec:2AgentRegret}--\ref{sec:ConservationVolume} respectively. 
%The results in \cite{Bailey2019Regret} related to (\ref{4}) and (\ref{5}) only applied to zero-sum games. 
Unlike standard analyses in online optimization, we prove our results for a generalized notion of learning rates. 
Specifically, we allow individual agents to use  different learning rates for each individual strategy. 
For instance, suppose an agent fundamentally believes that the strategy ``rock'' is the most important strategy in the game rock-paper-scissors. 
Then they may wish to use a larger learning rate for rock relative to scissors, e.g., a learning rate of ${\eta}_{rock}=100$ and ${\eta}_{scissors}=2$. 
In this case, if an agent observes a benefit of 1 for both rock and scissors, then the agent will increase their weight for rock by ${\eta}_{rock}\cdot 1 =100$ while only increasing their weight for scissors by ${\eta}_{scissors}\cdot 1 =2$. 
For a single agent, we do not see an immediate algorithmic benefit of using different learning rates and therefore make no suggestion for it in practice.
However, this generalization will be important for extending our results to multiagent systems in Section \ref{sec:Multi}.
We also remark that \cite{Bailey2019Regret} proves (\ref{1}) using a scalar learning rate and  (\ref{5}) in the setting of only zero-sum games using a scalar learning rate.

We begin by presenting Algorithm \ref{alg:2Agent} for alternating gradient descent between 2 agents.
In Algorithm \ref{alg:2Agent}, $D_{\eta}$ represents a diagonal matrix where the diagonal is populated by the vector of learning rates $\eta$.  
Similarly, $D_{\eta}Ay^{t-1}$ can be expressed by the Hadamard product $\eta \circ Ay^{t-1}$ indicating that the $i$th strategy is weighted according to $\eta_i$.  
However, for notation purposes, it will be simpler to work with the diagonal matrix $D_{\eta}$.
We also remark that for all of our analysis that $D_\eta$ can be replaced with an arbitrary positive-definite matrix.

\begin{varalgorithm}{2AltGD}
	\caption{2-Agent gradient descent with alternating updates}\label{alg:2Agent}
	\label{alg:euclid}
	\begin{algorithmic}[1]
		\Procedure{2AltGD}{$A,B,x^0,y^0,\eta,\gamma$}\Comment{Payoff matrices, initial strategies and learning rates}
		\For{\texttt{$t=1,...,T$}}
		\State $x^t:= x^{t-1} + D_{{\eta}}  Ay^{t-1}$ \Comment{Update strategies based on previous iteration} \label{line:agent1}
		\State $y^t:= y^{t-1} + D_{\gamma}  Bx^{{t}}$ \Comment{Update strategies based on current iteration}\label{line:agent2}
		\EndFor
		\EndProcedure
	\end{algorithmic}
\end{varalgorithm}

\begin{remark} If line \ref{line:agent2} of Algorithm \ref{alg:2Agent} is replaced with $x^{t-1}$ instead of $x^t$, then the algorithm is the normal implementation of gradient descent with simultaneous updates (Algorithm \ref{alg:GradientMulti}). 
\end{remark}

\subsection{$O\left(1/T\right)$ Time-Average Regret}\label{sec:2AgentRegret}

In traditional algorithmic settings, where agents update simultaneously, agent 1's regret with respect to a fixed strategy $x$ is defined as
\begin{align*}
\sum_{t=0}^T \langle x, Ay^t\rangle - \sum_{t=0}^T \langle x^t, Ay^t\rangle\tag{Standard Notion of Regret for Simultaneous Updates}
\end{align*}
i.e., the difference between the utility agent 1 would receive if the fixed strategy $x$ was played against $\{y^t\}_{t=0}^T$ and the utility agent 1 received by playing the sequence $\{x^t\}_{t=0}^T$.
Regret is the standard notion used to understand the performance of algorithms in repeated games and in online optimization in general. 
In the setting of bounded zero-sum games, it is well-known that the time-average of the strategies converges to the set of Nash equilibria whenever regret grows at  rate $o(T)$ (sublinearly).
Generally in online optimization, if regret grows at rate $o(T)$, then the time-average regret converges to zero implying that, on average, the algorithm performs as well as the fixed strategy $x$.

In the setting of alternating play where agents take turns updating, agent 2 plays the strategy $y^t$ twice -- once in the $t$th iteration when agent 2 updates ($x^t,y^t$) and once when agent 1 updates in the $(t+1)$th iteration ($x^{t+1},y^t$).  
As such, we update the notion of regret accordingly:
\begin{align*}
	\sum_{t=0}^T \langle 2x, Ay^t\rangle - \sum_{t=0}^T \langle x^{t+1}+x^t, Ay^t\rangle\tag{Regret After Agent 1 Updates}
\end{align*}

From an economic standpoint, it makes sense that agents would receive utility after each update. 
If agents only received utility after both agents updated, then the agent that updates last would decidedly have an advantage since they would see the other agent's strategy. 
As such, no rational agent would agree to take turns updating unless they receive utility every time they update. 
We remark that this notion of regret only captures agent 1's regret after  1 updates and is not sufficient on its own to guarantee time-average convergence. 
We discuss the implication of this definition at the end of this section.

\begin{theorem}\label{thm:regret}
	If agent 1 updates their strategies with Algorithm \ref{alg:2Agent} with an \underline{arbitrary} vector of fixed learning rates ${\eta}$, then agent 1's time-average regret with respect to an arbitrary fixed strategy $x$ after updating  in iteration $(T+1)$ is $O\left(1/T\right)$, regardless of how their opponent updates. 
	More specifically, agent 1's total regret is exactly
	\begin{align*}
		\left\langle x^0-2x, D^{-1}_{{\eta}}x^0\right\rangle + \left\langle 2x-x^{T+1},  	 D^{-1}_{{\eta}}x^{T+1} \right\rangle\leq \left\langle x^0-2x, D^{-1}_{{\eta}}x^0\right\rangle + \left\langle x,   D^{-1}_{{\eta}}x\right\rangle\in O(1).
	\end{align*}
\end{theorem}
\begin{proof}
	The total regret for agent 1 after agent 1 updates in iteration $(T+1)$ is
	\begin{align*}
			\sum_{t=0}^T\left\langle 2x-x^{t+1}-x^t,  Ay^t\right\rangle
		&= 
			\sum_{t=0}^T\left\langle 2x-x^{t+1}-x^t,  	  D^{-1}_{{\eta}}(x^{t+1}-x^t)\right\rangle\\
		&= 
			\sum_{t=0}^T\left(\left\langle x^t-2x, D^{-1}_{{\eta}}x^t\right\rangle - \left\langle x^{t+1}-2x,  	 D^{-1}_{{\eta}}x^{t+1}\right\rangle\right)\\
		&= 
			\left\langle x^0-2x, D^{-1}_{{\eta}}x^0\right\rangle + \left\langle 2x-x^{T+1},  	 D^{-1}_{{\eta}}x^{T+1} \right\rangle\\
		&\leq  
			\left\langle x^0-2x, D^{-1}_{{\eta}}x^0\right\rangle + \left\langle x,   D^{-1}_{{\eta}}x\right\rangle\in O(1)
	\end{align*}
	where the first equality follows from line \ref{line:agent1} of Algorithm \ref{alg:2Agent}, the second equality follows since $D_\eta^{-1}$ is symmetric, the third equality follows by canceling out terms from the telescopic sum, and where the inequality follows since the function $f(w):= \left \langle 2x-w,  D^{-1}_{{\eta}}w \right\rangle$ has a critical point at $w=x$, which corresponds to a global maximum since $D_\eta^{-1}$ is positive-definite. 
	Dividing the above equations by $T$ yields that the time-average regret is in $O\left(1/T\right)$.
\end{proof}

Theorem \ref{thm:regret} implies that agent 1's regret does not grow at all.  
This suggests that agent strategies will quickly converge to optimality in zero-sum games; we formally show this in Section \ref{sec:ZeroSum}. 
Interestingly, this result implies that we can compute agent 1's regret using very small amount of information.  Specifically, we only need to know agent 1's first and last strategy (with no information about agent 2) to compute their total regret.

While this bound on regret is incredibly powerful -- it holds regardless of how the opponent updates and for any learning rate -- the guarantee does not necessarily hold if regret is computed after agent 2 updates. 
As demonstrated in the proof of Proposition \ref{prop:BadRegret}, agent 2 can make their final strategy arbitrarily large in order to make agent 1's regret arbitrarily large. 
However, in practice, we do not necessarily expect agent 2 to play large strategies; for instance in Section \ref{sec:ZeroSum}, we show that $y^{T}$ is bounded when both agents use alternating gradient descent in zero-sum games.  This implies that agent 1 has bounded regret even when regret is computed after agent 2 updates (Corollary \ref{cor:ZSRegret}).

\begin{proposition}\label{prop:BadRegret}
	Suppose $A$ is invertible. If agent 1's regret is computed after agent 2 updates, then agent 1's regret with respect to $x$ can be made arbitrarily large if $A^{-1}(x-x^{T+1})\neq \vec{0}$.
\end{proposition}	

\begin{proof}
	After agent $2$ updates, agent 1's regret is given by 
	\begin{align*}
	&\sum_{t=0}^T\left\langle 2x-x^{t+1}-x^t,  Ay^t\right\rangle+\left\langle x-x^{T+1},  Ay^{T+1}\right\rangle
	\\=&\left\langle x^0-2x, D^{-1}_{{\eta}}x^0\right\rangle + \left\langle 2x-x^{T+1},  	 D^{-1}_{{\eta}}x^{T+1} \right\rangle+\left\langle x-x^{T+1},  Ay^{T+1}\right\rangle.
	\end{align*}
	Let $y^{T+1}=\lambda \cdot A^{-1}(x-x^{T+1})\neq \vec{0}$.  Then agent 1's regret after agent 2 updates approaches infinity as $\lambda\to \infty$. 
\end{proof}

\subsection{An Argument for Large Learning Rates}\label{sec:Large}

In most settings of online optimization, small learning rates are used to prove optimization guarantees. 
However, in this setting we actually show that a large learning rate yields stronger lower bounds on the utility gained. 

\begin{theorem}\label{thm:large1}
	Agent 1's total utility after updating in the $(T+1)th$ iteration is $\langle x^{T+1}, D_{\eta}^{-1}x^{T+1}\rangle-\langle x^{0}, D_{\eta}^{-1}x^{0}\rangle \geq -\langle x^{0}, D_{\eta}^{-1}x^{0}\rangle$.
\end{theorem}

\begin{proof}
Following identically to the proof of Theorem \ref{thm:regret},
	\begin{align*}
\sum_{t=0}^T\left\langle x^{t+1}+x^t,  Ay^t\right\rangle
&= 
\sum_{t=0}^T\left\langle x^{t+1}+x^t,  	  D^{-1}_{{\eta}}(x^{t+1}-x^t)\right\rangle\\
&= 
\sum_{t=0}^T\left(\left\langle x^{t+1},  	  D^{-1}_{{\eta}}x^{t+1}\right\rangle-\left\langle x^t,  	  D^{-1}_{{\eta}}x^t\right\rangle\right)\\
&= 
\langle x^{T+1}, D_{\eta}^{-1}x^{T+1}\rangle-\langle x^{0}, D_{\eta}^{-1}x^{0}\rangle \geq -\langle x^{0}, D_{\eta}^{-1}x^{0}\rangle.
\end{align*}
The lower bound follows since $D_{\eta}$ is positive-definite implying $\langle x^{T+1}, D_{\eta}^{-1}x^{T+1}\rangle\geq 0$.
\end{proof}

Recalling that $D_{\eta}$ is positive-definite, the bound $-\langle x^0, D_{\eta}^{-1} x^0\rangle <0$ and converges to $0$ as the learning rate grows large, i.e., {by using an arbitrarily large learning rate, an agent can guarantee that they lose arbitrarily little utility}. 
This is contrary to most online learning algorithms that suggest small, relatively unresponsive learning rates from agents. 
Admittedly, Theorem \ref{thm:large1} only provides a lower bound that depends on the learning rates and says little about the cumulative utility as a function of the learning rate $\eta$.  
However, in Theorem \ref{thm:large2}, we show that {an agent is better served with large learning rates when playing against an unresponsive agent}.

\begin{theorem}\label{thm:large2}
	If agent 1 is playing against an oblivious, non-equilibrating opponent -- i.e., if $\{y^t\}_{t=0}^\infty$ is independent of $\{x^t\}_{t=0}^\infty$ and $\sum_{t=0}^TAy^t\neq \vec{0}$ -- then agent 1 can make their utility arbitrarily high after updating in the $(T+1)th$ iteration by making $\eta$ arbitrarily high. 
\end{theorem}

\begin{proof}
	Agent 1's total utility is 
	\begin{align*}
		\langle x^{T+1}, D_{\eta}^{-1} x^{T+1} \rangle - \langle x^{0}, D_{\eta}^{-1} x^{0} \rangle 
		&=\left\langle x^0 + D_\eta\sum_{t=0}^T Ay^t, D_{\eta}^{-1} \left(x^0 + D_\eta\sum_{t=0}^T Ay^t\right) \right\rangle - \langle x^{0}, D_{\eta}^{-1} x^{0} \rangle \\
		&=2\left\langle x^0, \sum_{t=0}^T Ay^t\right \rangle+ \left\langle  D_\eta\sum_{t=0}^T Ay^t,\sum_{t=0}^T Ay^t\right\rangle\\
		&\to \infty \ as \ \eta\to \infty
	\end{align*}
	thereby completing the proof of the theorem. 
\end{proof}

\subsection{Self-Actualization}\label{sec:Actualization}

Next, we show that in order to maximize agent 1's regret for not playing $x\in {\cal X}$, Algorithm \ref{alg:2Agent} will actually force agent 1 to play $x$. 
We refer to this property as \textit{self-actualization}. 
Once agent 1 regrets not playing the strategy $x$ as much as possible, the agent will realize that strategy.

\begin{theorem}\label{thm:Actualization}
	Suppose agent 1 updates their strategies with Algorithm \ref{alg:2Agent}. 
	If the opponent's actions $\{y^0,...,y^T\}$ maximize agent 1's regret after agent 1 updates in the $(T+1)th$ iteration for not playing the fixed strategy $x$, then $x^{T+1}=x$. 
\end{theorem}

Ordinarily, we would have to be quite careful in making this claim and trying to prove it. 
Altering the sequence $\{y^t\}_{t=0}^T$ alters agent 1's sequence $\{x^t\}_{t=1}^T$ and there it seems difficult to explicitly give a sequence $\{y^t\}_{t=0}^T$ that maximize agent 1's regret. 
However, the proof of Theorem \ref{thm:regret} is quite strong -- the total regret  relies {only} on $x^0$ and $x^{T+1}$. 
The proof of Theorem \ref{thm:Actualization} follows immediately from Theorem \ref{thm:regret} since the upper bound was found using the unique optimizer $x^{T+1}=x$.

\subsection{Conservation of Volume in general-sum games}\label{sec:ConservationVolume}

In this section, we examine the volume expansion/contraction properties of Algorithm \ref{alg:2Agent}. 
Formally, let $V^0\subseteq {\cal X}\times {\cal Y}$ be a measurable set of initial conditions and let $V^t$ be the set obtained after updating every point in $V^{t-1}$ with Algorithm \ref{alg:2Agent} (see Figure \ref{fig:Volume}).  
Formally, $V^t=\bigcup_{\{x,y\}\in V^{t-1}}\left\{x+\eta A y, y +\gamma  B (x+\eta A y) \right\}$. 
We compare the volume of $V^0$ to $V^t$; specifically, we show that this volume is invariant.

\begin{figure}[H]
	\def\S{0.35}
	\centering
	\begin{subfigure}[b]{0.32\textwidth}
		\centering
		{\includegraphics[trim=0.9cm 0.5cm 1.75cm 2cm, clip=true,scale=\S]{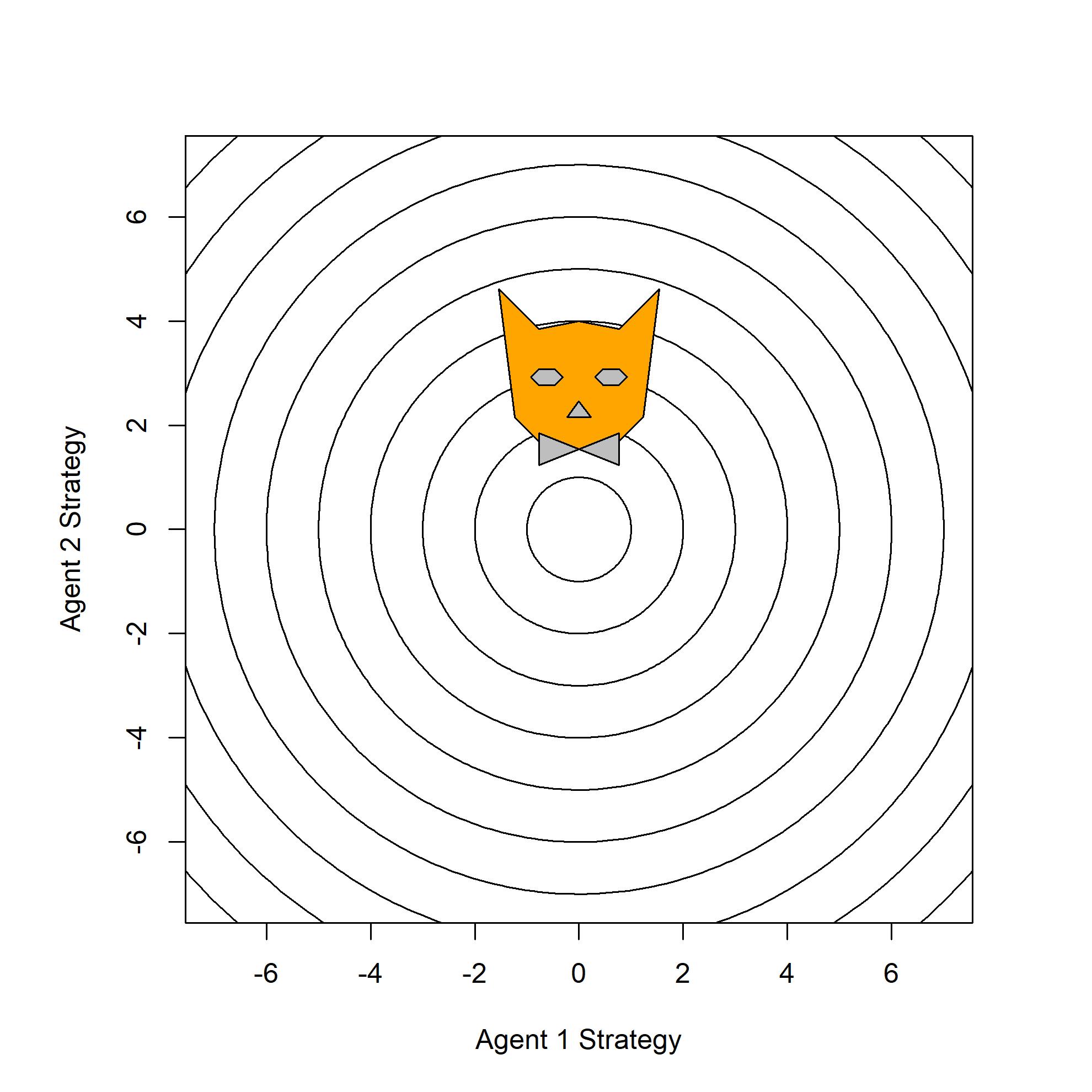}}
		\caption{After 1st Iteration: $V^1$}
	\end{subfigure}\hfill
	\begin{subfigure}[b]{0.33\textwidth}
		\centering	
		{\includegraphics[trim=0.9cm 0.5cm 1.75cm 2cm, clip=true,scale=\S]{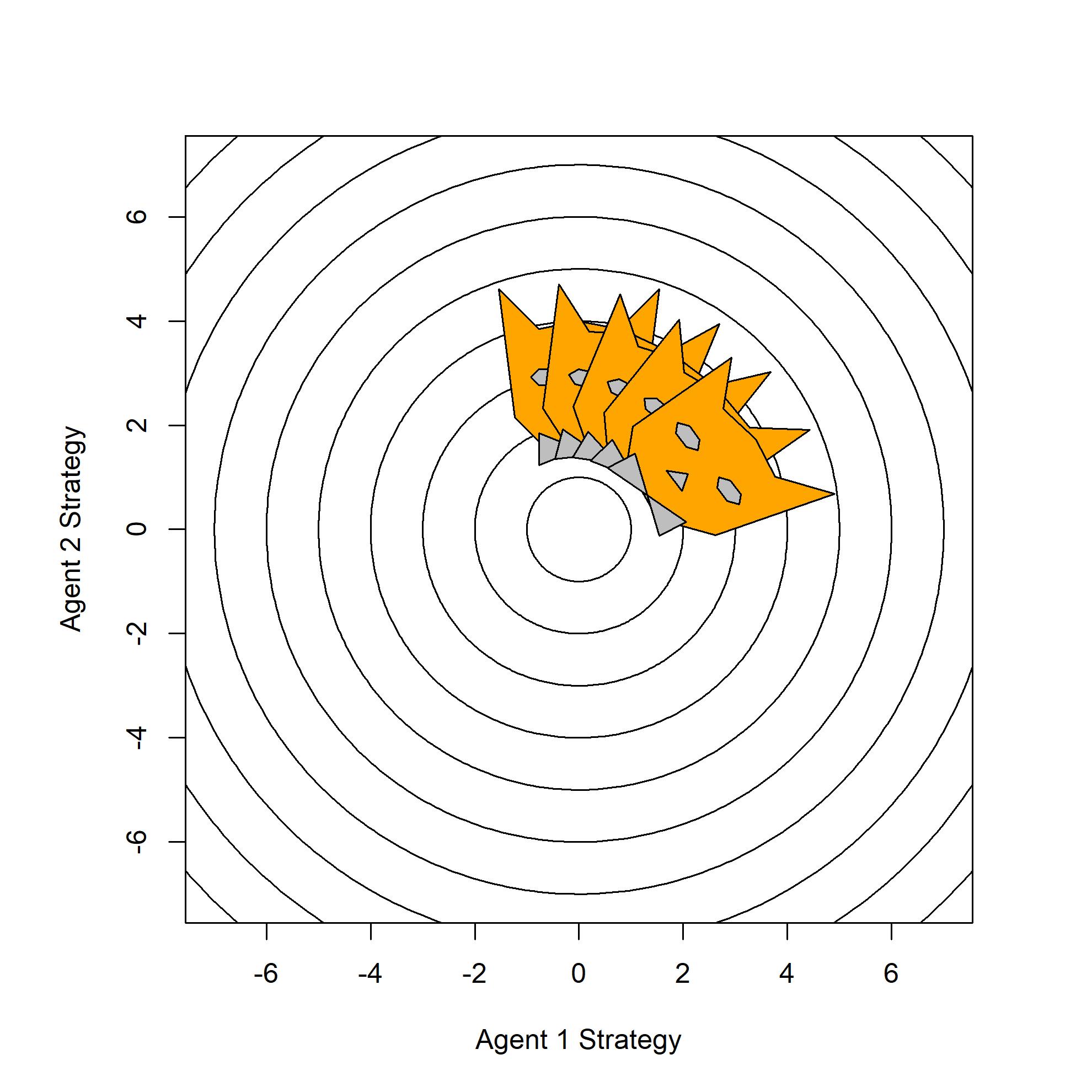}}
		\caption{After 5th Iteration: $V^1$ to $V^5$}
	\end{subfigure}\hfill
	\begin{subfigure}[b]{0.32\textwidth}
		\centering
		{\includegraphics[trim=0.9cm 0.5cm 1.75cm 2cm, clip=true,scale=\S]{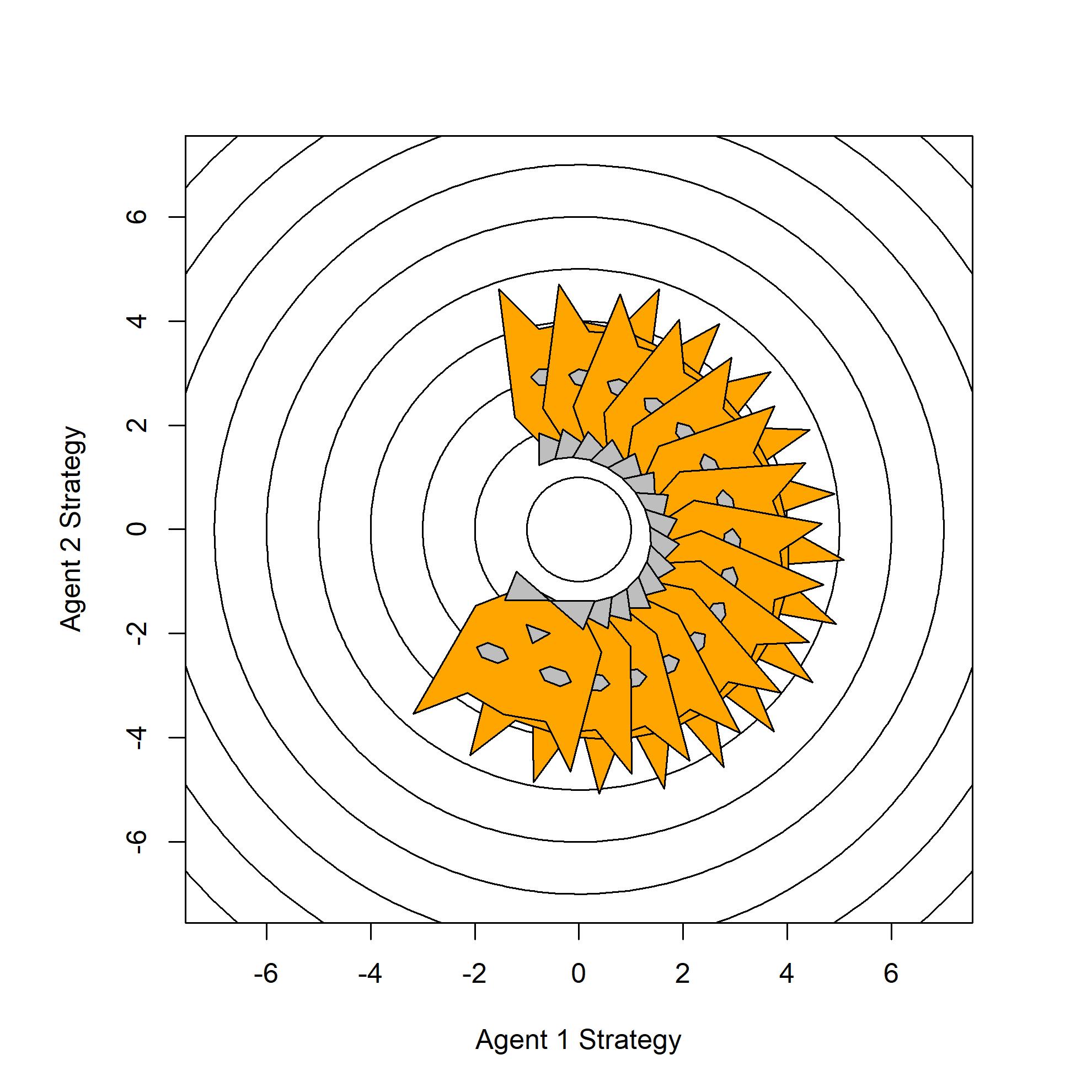}}
		\caption{After 15th Iteration: $V^1$ to $V^{15}$}
	\end{subfigure}
	\caption{Evolution of Algorithm \ref{alg:2Agent} on 4 sets of initial conditions (a cat, a pair of eyes, a mouth, and a bow-tie) in a zero-sum game with $A=-B=[1]$ and $\eta=\gamma=0.25$ after 1, 5, and 15 iterations respectively. In each iteration, every point in each set is updated according to alternating gradient descent and the position/shape changes but volume is preserved (Theorem \ref{thm:Volume}).}
	\label{fig:Volume}
\end{figure}

On its own, volume conservation is nice stability property due to its close connection with Lyapunov chaos. 
Lyapunov chaos refers to a phenomenon in dynamical systems where a small perturbation in initial conditions may result in arbitrarily different dynamical systems. 
Specifically, volume expansion implies that a small perturbation to the initial conditions can result in drastically different trajectories. 
Formally, let $V^0$ be a relatively small measurable set of initial conditions. 
If the volume of $V^t$ goes to infinity, then there exists an iteration $T$ and two points $(x^T,y^T)$ and $(\bar{x}^T,\bar{y}^T)$ that are arbitrarily far apart.  
However, by definition of $V^t$, $(x^T,y^T)$ and $(\bar{x}^T,\bar{y}^T)$ evolve from $(x^0,y^0)\in V^0$ and $(\bar{x}^0,\bar{y}^0)\in V^0$.
This implies the two points, despite being close together initially, will diverge from one another over time.  

We show that alternating gradient descent is volume preserving in general-sum 2-agent games. 

\begin{theorem}\label{thm:Volume}
	Algorithm \ref{alg:2Agent} is volume preserving for any measurable set of initial conditions. 
\end{theorem}

\begin{proof}
	Algorithm \ref{alg:2Agent} can be expressed as the two separate updates below.
	\begin{align*}
		\left[\begin{array}{c} x^{t+1}\\ y^t\end{array}\right] &\gets \left[\begin{array}{c} x^{t}+D_{\eta}Ay^t\\ y^t\end{array}\right] \tag{Line \ref{line:agent1} of Algorithm \ref{alg:2Agent}}\\
		\left[\begin{array}{c} x^{t+1}\\ y^{t+1}\end{array}\right] &\gets \left[\begin{array}{c} x^{t+1}\\ y^{t}+D_{\gamma}Bx^{t+1}\end{array}\right] \tag{Line \ref{line:agent2} of Algorithm \ref{alg:2Agent}}
	\end{align*}
	To show that the combined updates preserve volume, it suffices to show that each individual update preserves volume.  
	Thus, it suffices to show that the absolute value of the determinant of the Jacobian for each update is 1 \cite[Theorem 7.26]{rudin1987real}. 
	The Jacobians for the updates are
	\begin{align*}
	J_1=\left[\begin{array}{c c} I_{\cal X} & D_{\eta} A\\ 0 & I_{\cal Y}\end{array}\right]  \tag{Line \ref{line:agent1} Jacobian}\\
	J_2=\left[\begin{array}{c c} I_{\cal X} & 0\\ D_{\gamma}B & I_{\cal Y}\end{array}\right]  \tag{Line \ref{line:agent2} Jacobian}
	\end{align*}
	where $I_{\cal X}$ and $I_{\cal Y}$ are the identity matrices with the same dimension as ${\cal X}$ and ${\cal Y}$ respectively. 
	
	Since both Jacobians are block triangular with zeros on the subdiagonal and superdiagonal respectively, their corresponding determinants are $\det(J_i)=\det(I_{\cal X})\cdot \det(I_{\cal Y})=1$ and therefore Algorithm \ref{alg:2Agent} preserves volume when updating a measurable set of strategies thereby completing the proof.
\end{proof}

\begin{remark}
	Volume conservation holds even if agents' learning rates change overtime ($\{\eta_t\}_{t=0}^T\}$) since the determinant of the Jacobian is independent of $\eta$.  
\end{remark}

Thus, alternating gradient descent preserves volume.  
This is in contrast to the standard implementation of gradient descent (Algorithm \ref{alg:GradientMulti}) where volume expands in zero-sum games (see  \cite{Marco2020Chaos} and Figure \ref{fig:Volume2}).

\begin{figure}[H]
	\def\S{0.35}
	\centering
	\begin{subfigure}[b]{0.45\textwidth}
		\centering
		{\includegraphics[trim=0.9cm 0.5cm 1.75cm 2cm, clip=true,scale=\S]{AltGD15.jpg}}
		\caption{Volume Conservation of Algorithm \ref{alg:2Agent}.}
	\end{subfigure}\hfill
	\begin{subfigure}[b]{0.45\textwidth}
		\centering
		{\includegraphics[trim=0.9cm 0.5cm 1.75cm 2cm, clip=true,scale=\S]{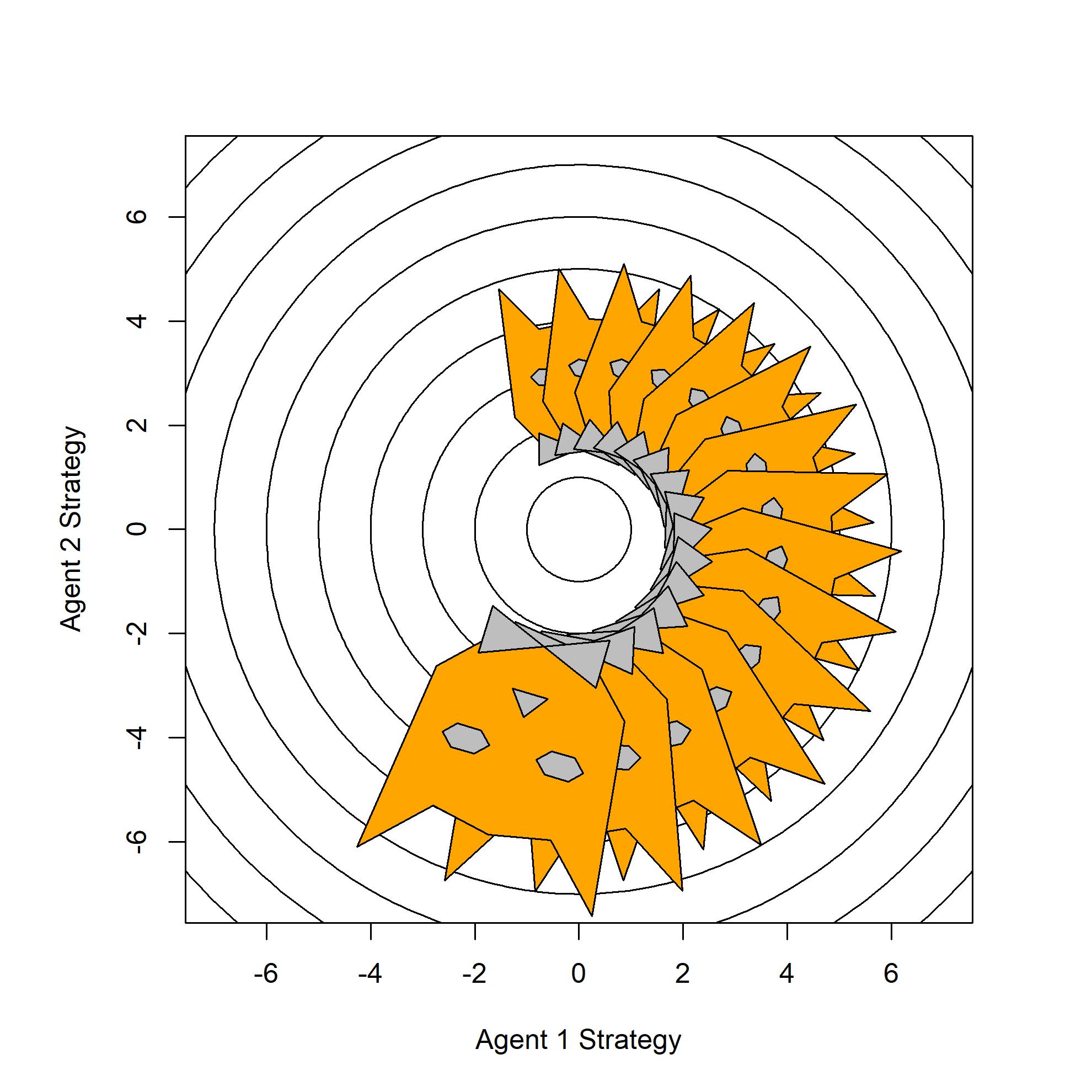}}
		\caption{Volume Expansion of Algorithm \ref{alg:GradientMulti}.}
	\end{subfigure}
	\caption{Evolution of Algorithms \ref{alg:2Agent} and \ref{alg:GradientMulti} respectively after 15 iterations.  Both algorithms start with the same set of initial strategies. Volume is preserved with when updating with Algorithm \ref{alg:2Agent} but expands when using Algorithm \ref{alg:GradientMulti}.}
	\label{fig:Volume2}
\end{figure}

Regrettably however, volume conservation is insufficient to avoid Lyapunov chaos; 
in Lemma \ref{prop:chaos2}, we show that two points can still move arbitrarily far apart in the setting of a coordination game as depicted in Figure \ref{fig:Volume3}. 

\begin{figure}[H]
	\def\S{0.35}
	\centering
	{\includegraphics[trim=0.9cm 0.5cm 1.75cm 2cm, clip=true,scale=\S]{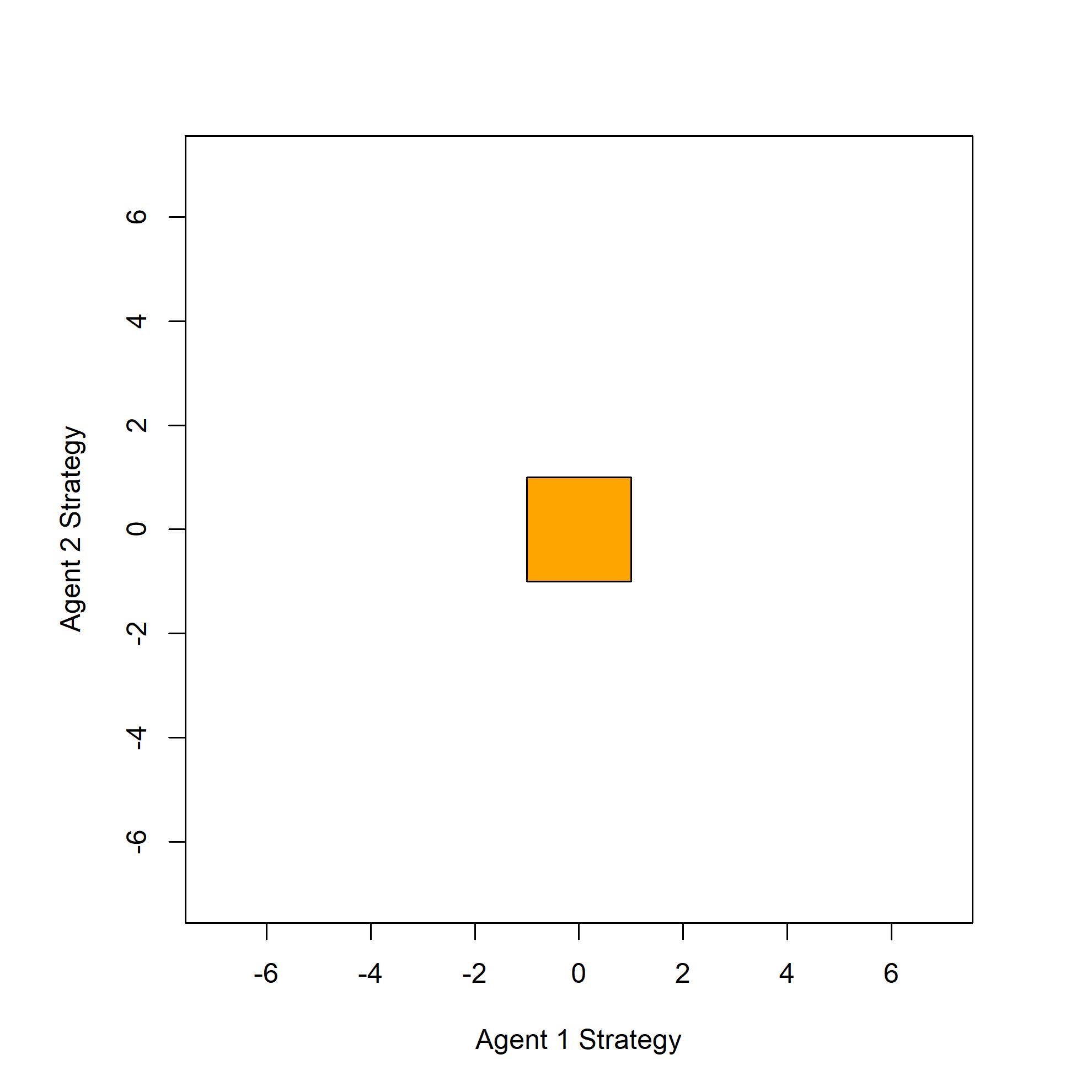}}\hfill
	{\includegraphics[trim=0.9cm 0.5cm 1.75cm 2cm, clip=true,scale=\S]{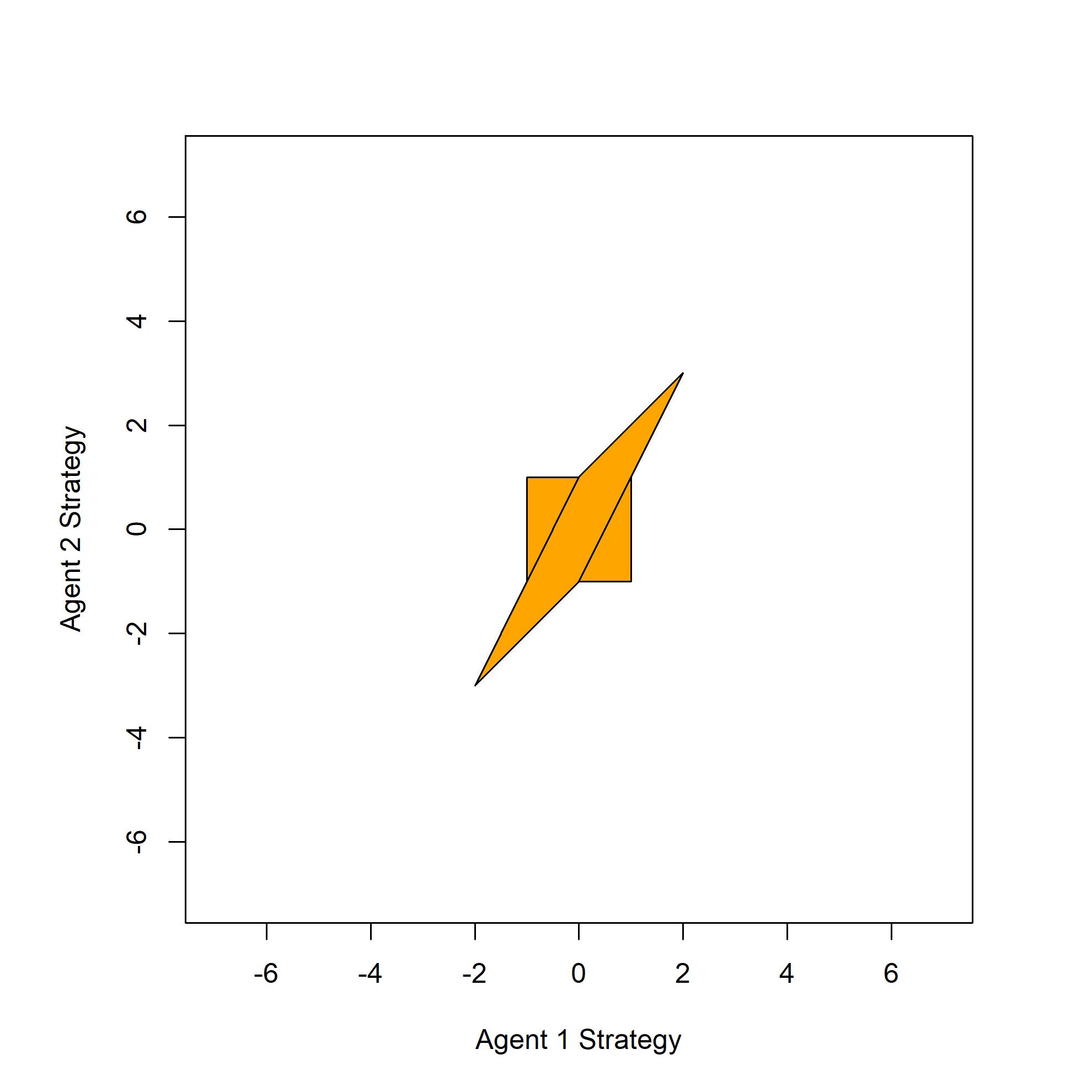}}\hfill
	{\includegraphics[trim=0.9cm 0.5cm 1.75cm 2cm, clip=true,scale=\S]{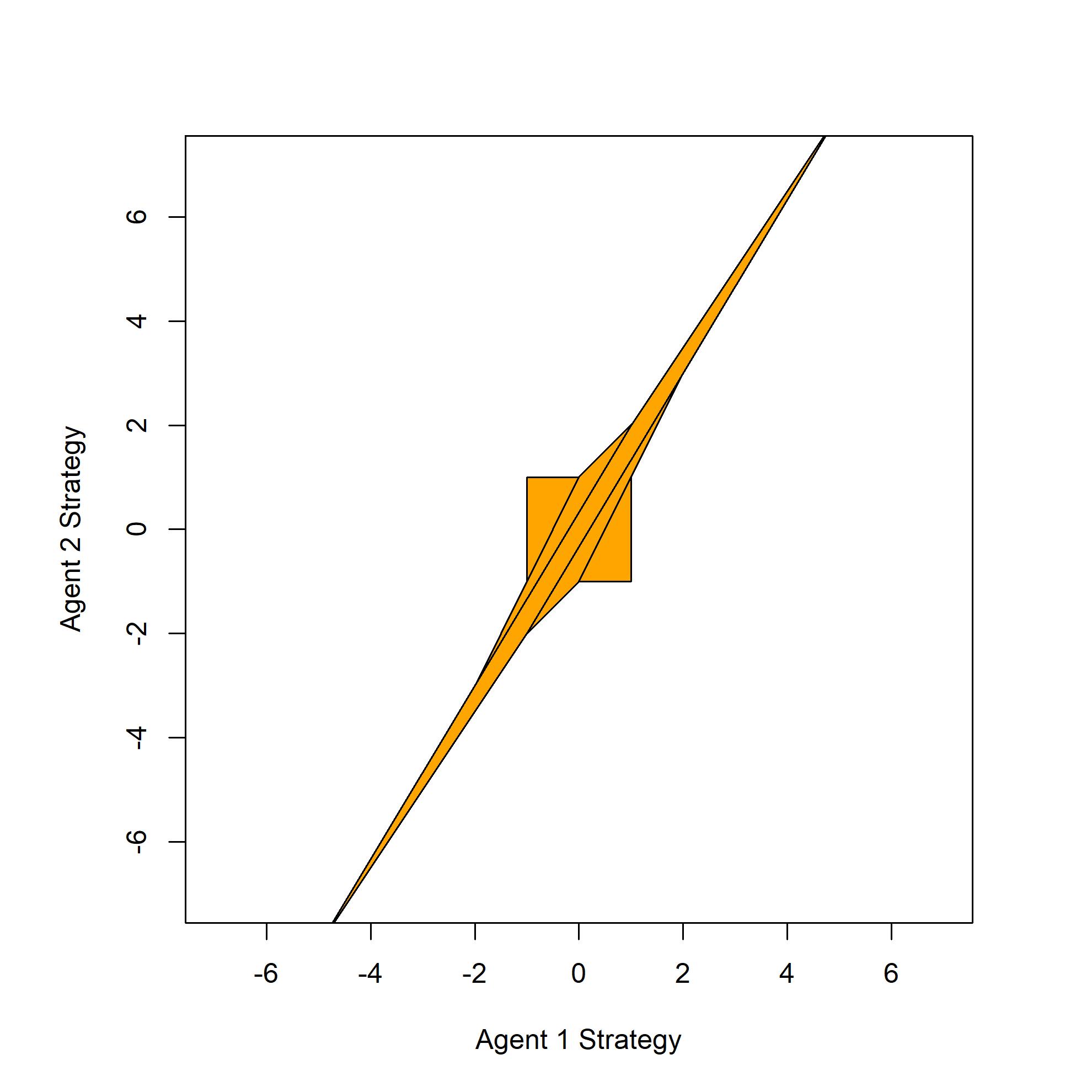}}
	\caption{Evolution of $[-1,1]\times [-1,1]$ using $A=B=[1]$ and $\eta=\gamma=1$ after 1, 2, and 3 iterations respectively.  Volume is preserved with Algorithm \ref{alg:2Agent} but the diameter grows exponentially.}
	\label{fig:Volume3}
\end{figure}

\begin{restatable}[]{lemma}{Chaos}\label{prop:chaos2}
	Let  $A=B=[1]$ and $V^0=[-1,1]\times [-1,1]$ with $\eta=\gamma =1$. The volume of $V^t$ is 4 while the diameter of $V^t$ is in $\Theta(\phi^{2t})$ where $\phi=(1+\sqrt{5})/2$ is the golden-ratio. 
\end{restatable}
We show that the extreme points of $V^t$ are $(\pm (F_{2t-2}, F_{2t-1}), \pm (F_{2t+1}, F_{2t+2})$ where $F_k$ is the $k$th Fibonacci number.  
The result regarding the diameter immediately follows since it is well-known that $F_k\in \Theta(\phi^k)$.
Finally, the volume of $V^0$ is 4, implying the volume of $V^t$ is also 4 since volume is invariant (Theorem \ref{thm:Volume}). 
The full proof appears in Appendix \ref{app:Chaos}.

\section{2-Agent Positive-Negative Definite (Zero-Sum) Games}\label{sec:ZeroSum}

In this section, we introduce a new class of games that includes all zero-sum games and show that Algorithm \ref{alg:2Agent} results in strategies that are bounded (Theorem \ref{thm:BoundedOrbits}), are Poincar\'{e} recurrent (Theorem \ref{thm:Recurrence}), and have $O\left( 1/T\right)$ time-average convergence to the set of Nash equilibria (Theorem \ref{thm:converge}). 
Specifically, we study a generalization of zero-sum games that allows each agent to multiply their payoff matrices by arbitrary positive definite matrices $P$ and $Q$ respectively, i.e., $P=P^\intercal$ and $\langle x,Px\rangle >0$ for all $x\in \mathbb{R}^{S_i}$. 

\begin{equation}
\begin{aligned}
\max_{x\in {\cal X}} \ &\langle x, P A y\rangle \\
\max_{y\in {\cal Y}} \ &\langle y, -Q A^\intercal x\rangle \\
\end{aligned}\tag{Positive-Negative Definite Game}\label{eqn:posneg}
\end{equation}

We remark that recurrence (Theorem \ref{thm:Recurrence}) and bounded orbits (Theorem \ref{thm:BoundedOrbits}) were shown for zero-sum games (without positive definite transformations) with a scalar learning rate in \cite{Bailey2019Regret}. 	
Unlike the results for regret in Theorem \ref{thm:regret}, arbitrary learning rates are not allowed -- to obtain $O\left( 1/T\right)$ time-average convergence, the learning rates must be sufficiently small. 
Importantly, we show that Algorithm \ref{alg:2Agent} allows four times larger learning rates than required for optimistic gradient descent.

\subsection{Importance of Positive-Negative Definite Games}

Zero-sum games are only a measure zero set of positive-negative definite games and therefore our results drastically expand the applications of learning algorithms.
This is particularly important for many economic settings where the underlying games are somewhat adversarial but not necessarily zero-sum. 
In such settings, it is currently unknown whether results for zero-sum games extend and thus, the best known for an algorithm in a similar setting is poly$(\log(T))$ time-average convergence to the set of coarse correlated equilibria \cite{daskalakis2021nearoptimal} -- a class of equilibria significantly less important than the set of Nash equilibria. 
We introduce techniques to show that Algorithm \ref{alg:2Agent} results in $O(1/T)$ time-average convergence to the set of Nash equilibria (Theorem \ref{thm:converge}) in this setting. 
We remark that the proof techniques we introduce can likely be used to extend many results for zero-sum games to positive-definite transformations of zero-sum games for other algorithms, e.g., optimistic gradient descent.

Unlike zero-sum games, in (\ref{eqn:posneg}) agent 1's utility function uncovers no information about agent 2's utility function.  In contrast, in a zero-sum game, agent 1 always has knowledge of agent 2's payout and can directly compute the set of Nash equilibria as a result.  
As shown in Proposition \ref{prop:NoNash}, it is impossible for agent 1 to independently determine a Nash equilibrium in a positive-negative definite game.

\begin{proposition}\label{prop:NoNash}
	Unlike zero-sum games, agent 1 cannot determine the set of Nash equilibria with access only to agent 1's payoff matrix in (\ref{eqn:posneg}). 
\end{proposition}

\begin{proof}
	To prove the proposition, we give two different games $\{P_1,Q_1,A_1\}$ and $\{P_2,Q_2,A_2\}$ where agent 1 has the same payoff matrix in both games ($P_1A_1=P_2A_2$) but where agent 1's set of Nash equilibria are different in each game. 
	
	\begin{align*}
	P=\left[\begin{array}{r r} 1 & 0 \\ 0 & 1 \end{array}\right], Q=\left[\begin{array}{r r} 1 & 0 \\ 0 & 1 \end{array}\right], A=\left[\begin{array}{r r} 1 & -1 \\ -1 & 1 \end{array}\right]\tag{Matrices for First Game} 	
	\end{align*}
	With respect to this game, $PA=\left[\begin{array}{r r} 1 & -1 \\ -1 & 1 \end{array}\right]$ implying agent 2's set of Nash equilibria is $\{y\in \mathbb{R}^2: y_1=y_2\}$. 
	Similarly, $-QA^\intercal =\left[\begin{array}{r r} -1 & 1 \\ 1 & -1 \end{array}\right]$ implying agent 1's set of Nash equilibria is $\{x\in \mathbb{R}^2: x_1=x_2\}$. 
	
	\begin{align*}
	P=\left[\begin{array}{r r} 2 & 0 \\ 0 & 1 \end{array}\right], Q=\left[\begin{array}{r r} 1 & 0 \\ 0 & 1 \end{array}\right],   A=\left[\begin{array}{r r} 1/2 & -1/2 \\ -1 & 1 \end{array}\right]\tag{Matrices for Second Game} 	
	\end{align*}
	With respect to this game, $PA=\left[\begin{array}{r r} 1 & -1 \\ -1 & 1 \end{array}\right]$ and agent 2's Nash equilibria are unchanged. 
	However, $-QA^\intercal =\left[\begin{array}{r r} -1/2 & 1 \\ 1/2 & -1 \end{array}\right]$ implying agent 1's set of Nash equilibria are $\{x\in \mathbb{R}^2: x_1=2x_2\}$. 
\end{proof}

\begin{remark}
	The game introduced in the proof of Proposition \ref{prop:NoNash} is necessarily degenerate; since $\vec{0}$ is always a Nash equilibrium, for two games to have a different set of Nash equilibria one game must have multiple Nash equilibria. 
	In Section \ref{sec:bilinear}, we extend our results to a generalization of bimatrix games that allows for an arbitrary unique Nash equilibrium.  
	It is then straightforward to extend Proposition \ref{prop:NoNash} using two non-degenerate games. 
\end{remark}

\subsection{Using the Correct Basis}

The adversarial nature of (\ref{eqn:posneg}) is better revealed when examining the game in the bases induced by the transformations $P$ and $Q$.  
As such, we introduce the notion of a weighted normal to simply our proofs.

\begin{definition}\label{def:Weighted}
	Let $W$ be a positive definite matrix ($W=W^\intercal$ and $x^\intercal Wx >0$ for all $x\neq 0$).  Then the weighted-norm of the vector $x$ with respect to $W$ is $\lVert x \rVert_W=\sqrt{\langle x, Wx\rangle}$.
\end{definition}
	
	Weighted norms are often used in physics and dynamical systems to understand movement with respect to a non-standard set of basic vectors. 
	While the euclidean norm, $||\cdot||$, is well-suited when understanding systems defined by the standard set of basic vectors -- the columns of an identity matrix -- the dynamics of (\ref{eqn:posneg}) are best understood in the vector spaces induced by $P^{-1}$ and $Q^{-1}$.
	In addition, it will be useful to relate the standard Euclidean norm to the weighted-norm via the following lemma.

	\begin{lemma}\label{lem:norm}
		Suppose $W$ is positive-definite.  Then $\lVert x \rVert \leq \left\lVert W^{\frac{1}{2}}\right\rVert\cdot\lVert x\rVert_{W^{-1}}$
	\end{lemma}
\begin{proof}
	First, observe that $\left\lVert W^{-\frac{1}{2}}x \right\rVert=\sqrt{\left\langle W^{-\frac{1}{2}}x,W^{-\frac{1}{2}}x\right\rangle}=\sqrt{\left\langle x,W^{-1}x\right\rangle}=\lVert x\rVert_{W^{-1}}$ since $W=W^\intercal$.
	Therefore,
	\begin{align*}
		\lVert x \rVert 
		= \left\lVert W^{\frac{1}{2}}W^{-\frac{1}{2}}x \right\rVert
		\leq \left\lVert W^{\frac{1}{2}} \right\rVert\left\lVert W^{-\frac{1}{2}}x \right\rVert = \left\lVert W^{\frac{1}{2}}\right\rVert\cdot\lVert x\rVert_{W^{-1}}
	\end{align*}	
	where the inequality follows by definition of the matrix norm $\lVert W^{\frac{1}{2}} \rVert = \max_y \frac{\lVert W^{\frac{1}{2}}y\rVert}{\lVert y\rVert}\geq \frac{\lVert W^{\frac{1}{2}}z\rVert}{\lVert z\rVert}$ with $z=W^{-\frac{1}{2}}x$.
\end{proof}

\subsection{Conservation of Energy}\label{sec:ZSenergy}

\begin{figure}[H]
	\def\S{0.4}
	\centering
	\begin{subfigure}[b]{0.49\textwidth}
		\centering
		{\includegraphics[trim=0.9cm 0.5cm 1.75cm 2cm, clip=true,scale=\S]{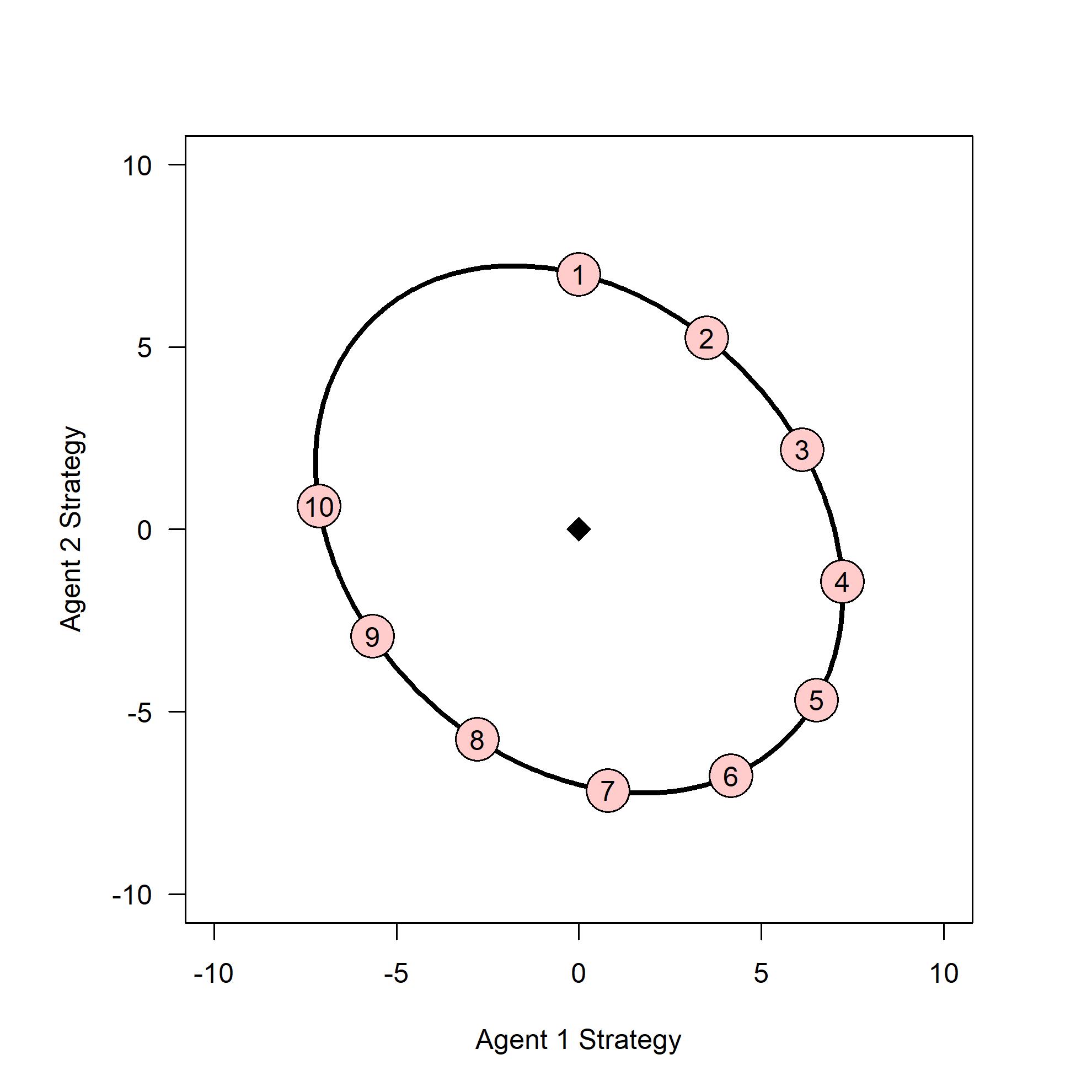}}
		\caption{Invariant Energy Function for the Zero-Sum Game $A=-B=[1]$.}
	\end{subfigure}\hfill
	\begin{subfigure}[b]{0.49\textwidth}
		\centering
		{\includegraphics[trim=0.9cm 0.5cm 1.75cm 2cm, clip=true,scale=\S]{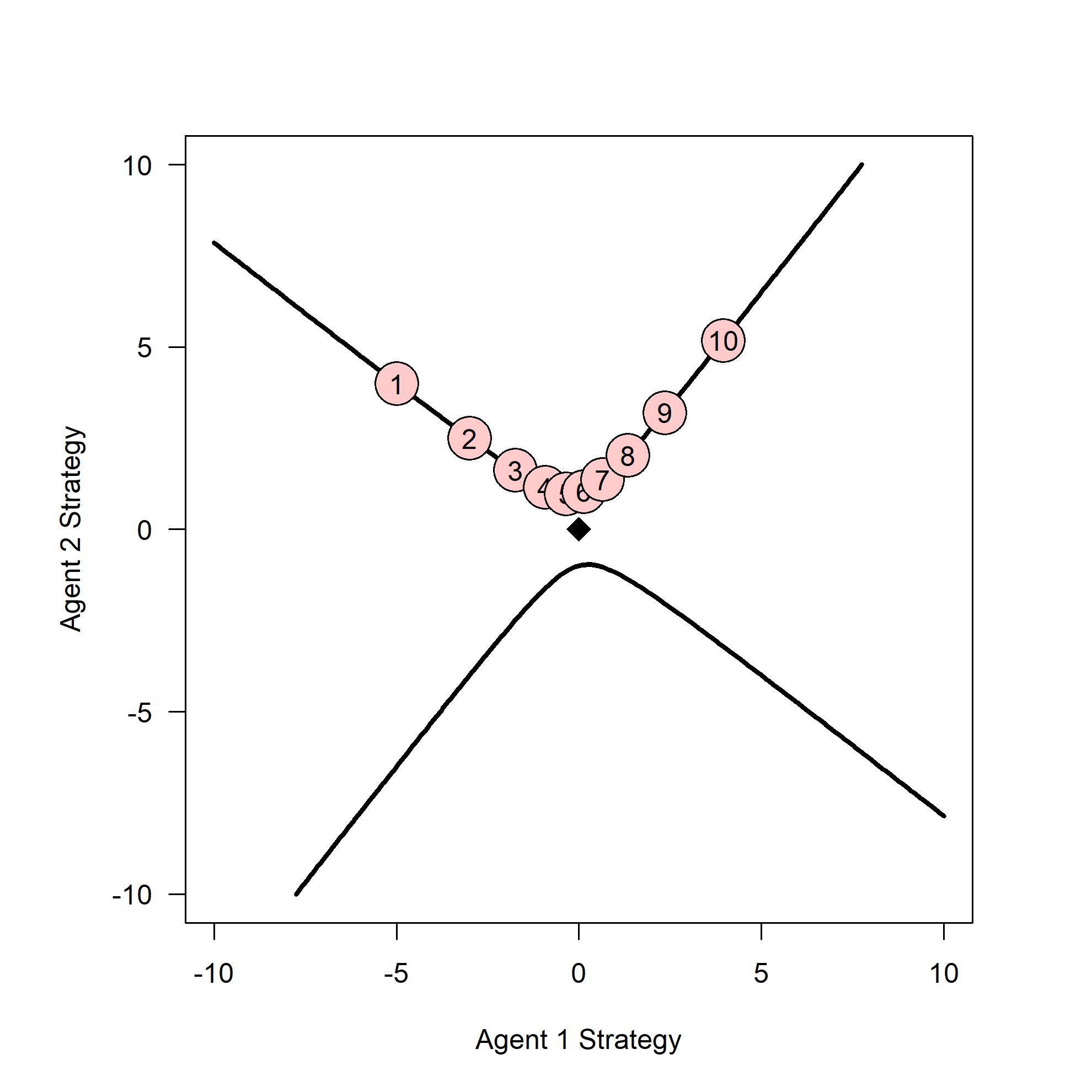}}
		\caption{Invariant Energy Function for the Coordination Game $A=B=[1]$.}
	\end{subfigure}	
	\caption{Evolution of Algorithm \ref{alg:2Agent} on the zero-sum game with $A=-B=[1]$ and the coordination game with $A=B=[1]$ respectively.  Both games use learning rates $\eta=\gamma=0.5$ and the combined strategies after $i$ iterations are marked by a red circle with the number $i$. The strategies move along the invariant energy functions given by Theorems \ref{thm:energy2} and \ref{thm:energy3} as marked by the black curve. }
	\label{fig:Energy}
\end{figure}

In this section, we show a strong stability condition of Algorithm \ref{alg:2Agent}; 
despite the algorithm being discrete, the updates all belong to a continuous, second-degree polynomial function -- an invariant ``energy function'' as depicted in Figure \ref{fig:Energy}. 
This energy function is a close perturbation of the energy found in \cite{Bailey19Hamiltonian} for zero-sum and coordination games in the continuous-time variant of gradient descent.

\begin{theorem}\label{thm:energy2}
	Suppose $P$ and $Q$ commute with $D_{\eta}$ and $D_{\gamma}$ respectively. 
	Then the perturbed energy $\left\lVert x^t\right\rVert^2_{P^{-1}D_{{\eta}}^{-1}} + \left\lVert y^t\right\rVert^2_{Q^{-1}D_{\gamma}^{-1}} + \langle x^t, Ay^t\rangle$ is invariant when agents play (\ref{eqn:posneg}) and update their strategies with Algorithm \ref{alg:2Agent}. 
\end{theorem}

We remark that the condition that $P$ and $D_{\eta}$ is not restrictive; it is trivially satisfied in traditional setting of online optimization where an agent uses a single learning rate for all strategies implying $D_{\eta}$ is a multiple of the identity matrix. 

\begin{proof}[Proof of Theorem \ref{thm:energy2}.]
	By the update rule given by line \ref{line:agent1} in Algorithm \ref{alg:2Agent}, 
	\begin{align*}
			\left\langle x^{t+1}+x^t, A y^t \right \rangle
		&=  \left\langle x^{t+1}+x^t,P^{-1}D_{\eta}^{-1}\left( x^{t+1}-x^t\right) \right \rangle
		= \left\lVert x^{t+1}\right\rVert^2_{P^{-1}D_{\eta}^{-1}}-\left\lVert x^{t}\right\rVert^2_{P^{-1}D_{\eta}^{-1}}
	\end{align*}
	since $P^{-1}$ and $D_\eta^{-1}$ are both positive definite and commute.
	Similarly, by line \ref{line:agent2} of Algorithm \ref{alg:2Agent}, 
	\begin{align*}
		\left\langle y^{t+1}+y^t ,-A^\intercal x^{t+1}  \right \rangle
		&=  \left\langle y^{t+1}+y^t, Q^{-1}D_{\gamma}^{-1}\left(y^{t+1}-y^t \right) \right \rangle
		=  \left\lVert y^{t+1}\right\rVert^2_{Q^{-1}D_{\gamma}^{-1}}-\left\lVert y^t\right\rVert^2_{Q^{-1}D_{\gamma}^{-1}} 
	\end{align*}
	Adding together both equalities and re-arranging terms yields
	\begin{align*}
		\left\lVert x^{t+1}\right\rVert^2_{P^{-1}D_{\eta}^{-1}} +\left\lVert y^{t+1}\right\rVert^2_{Q^{-1}D_{\gamma}^{-1}} + \langle x^{t+1}, Ay^{t+1}\rangle=
		\left\lVert x^{t}\right\rVert^2_{P^{-1}D_{\eta}^{-1}} + \left\lVert y^t\right\rVert^2_{Q^{-1}D_{\gamma}^{-1}}  + \langle x^t, Ay^t\rangle
 	\end{align*}
	thereby completing the proof of the theorem. 
\end{proof}
	
	\subsubsection{Energy in Positive-Positive Definite (Coordination) Games}\label{sec:COORDenergy}
	
	For completeness, we also give the energy function for positive-definite transformations of coordination games ($B=A^\intercal$).

	\begin{equation}
	\begin{aligned}
	\max_{x\in {\cal X}} \ &\langle x, P A y\rangle \\
	\max_{y\in {\cal Y}} \ &\langle y, Q A^\intercal x\rangle \\
	\end{aligned}\tag{Positive-Positive Definite Game}\label{eqn:pospos}
	\end{equation}
	
	\begin{theorem}\label{thm:energy3}
		Suppose $P$ and $Q$ commute with $D_{\eta}$ and $D_{\gamma}$ respectively. 
		Then the perturbed energy $\left\lVert x^t\right\rVert^2_{P^{-1}D_{{\eta}}^{-1}} - \left\lVert y^t\right\rVert^2_{Q^{-1}D_{\gamma}^{-1}} + \langle x^t, Ay^t\rangle$ is invariant when agents play a \ref{eqn:pospos} and update their strategies with Algorithm \ref{alg:2Agent}. 
	\end{theorem}
	
	The proof follows identically to the proof of Theorem \ref{thm:energy2} after adding together $\langle x^{t+1}, Ay^t\rangle$ and $-\langle y^{t+1}+y^t, A^\intercal x^{t+1}\rangle$.

\subsection{Bounded Orbits and Recurrence}

As shown in Figure \ref{fig:Energy}, in (\ref{eqn:posneg}) the strategies appear like they will cycle -- or at least will come close to cycling.  
In dynamics, this property is captured by Poincar\'{e} recurrence.

\begin{theorem}[Poincar\'{e} recurrence]\label{thm:Recurrence}
	Suppose $P$ and $Q$ commute with $D_{\eta}$ and $D_{\gamma}$ respectively and $\left\lVert P^{\frac{1}{2}}D_{\eta}^{\frac{1}{2}}\right\rVert\cdot\left\lVert Q^{\frac{1}{2}}D_{\gamma}^{\frac{1}{2}}\right\rVert< \frac{2}{||A||}$ when updating both agents' strategies with Algorithm \ref{alg:2Agent} in (\ref{eqn:posneg}). 
	For almost every initial condition $(x^0,y^0)$, there exists an increasing sequence of iterations $t_n$ such that $(x^{t_n},y^{t_n})\to (x^0,y^0)$.
\end{theorem}

Once again, the condition that $D_{\eta}$ and $P$ commute is naturally satisfies in standard applications. 

Poincar\'{e} recurrence guarantees that a system will come arbitrarily close to its initial conditions infinitely often.  
Informally, we think of this as cycling -- if our learning algorithm ever returns exactly to its initial condition, then the subsequent iterations will follow the prior iterations. 
By \cite{Poincare1890, barreira}, to formally show recurrence, it suffices to show that the updates are bounded and that the update rule preserves volume (Theorem \ref{thm:Volume}).
Thus, to complete the proof of Theorem \ref{thm:Recurrence}, it remains to show that $\{x^t,y^t\}_{t=0}^\infty$ is bounded.

\begin{theorem}\label{thm:BoundedOrbits}
	Suppose $P$ and $Q$ commute with $D_{\eta}$ and $D_{\gamma}$ respectively and $\left\lVert P^{\frac{1}{2}}D_{\eta}^{\frac{1}{2}}\right\rVert\cdot\left\lVert Q^{\frac{1}{2}}D_{\gamma}^{\frac{1}{2}}\right\rVert< \frac{2}{||A||}$ when updating both agents' strategies with Algorithm \ref{alg:2Agent} in (\ref{eqn:posneg}). 
	Then the agent strategies $\{x^t,y^t\}_{t=0}^\infty$ are bounded. Specifically,  
		\begin{align}
	\left\lVert x^t\right\rVert^2_{P^{-1}D_{{\eta}}^{-1}} + \left\lVert y^t\right\rVert^2_{Q^{-1}D_{\gamma}^{-1}}
	&\leq \frac{\left\lVert x^0\right\rVert^2_{P^{-1}D_{{\eta}}^{-1}} + \left\lVert y^0\right\rVert^2_{Q^{-1}D_{\gamma}^{-1}}+\left\langle x^0, Ay^0\right\rangle}
	{1-\frac{\left\lVert A\right\rVert\cdot\left\lVert P^{\frac{1}{2}}D_{\eta}^{\frac{1}{2}}\right\rVert\cdot\left\lVert Q^{\frac{1}{2}}D_{\gamma}^{\frac{1}{2}}\right\rVert}{2}}.
	\end{align}
\end{theorem}

\begin{proof}
	By Theorem \ref{thm:energy2}, energy is preserved and,
	\begin{align*}
	\left\lVert x^t\right\rVert^2_{P^{-1}D_{{\eta}}^{-1}} + \left\lVert y^t\right\rVert^2_{Q^{-1}D_{\gamma}^{-1}}=\left\lVert x^0\right\rVert^2_{P^{-1}D_{{\eta}}^{-1}} + \left\lVert y^0\right\rVert^2_{Q^{-1}D_{\gamma}^{-1}}+\langle x^0,Ay^0\rangle - \langle x^t,Ay^t\rangle.
	\end{align*}	
	
	Next, observe that
	\begin{align*}
		-\left\langle x^t, Ay^t\right\rangle 
		&\leq \left\lVert x^t\right\rVert\cdot\left\lVert Ay^t\right\rVert\\
		&\leq \left\lVert A\right\rVert\cdot\left\lVert x^t\right\rVert\cdot\left\lVert y^t\right\rVert\\
		&\leq \left\lVert A\right\rVert\cdot\left\lVert P^{\frac{1}{2}}D_{\eta}^{\frac{1}{2}}\right\rVert\cdot\left\lVert Q^{\frac{1}{2}}D_{\gamma}^{\frac{1}{2}}\right\rVert\cdot\left\lVert x^t\right\rVert_{P^{-1}D_{\eta}^{-1}}\cdot\left\lVert y^t\right\rVert_{Q^{-1}D_{\gamma}^{-1}} \\
		&\leq  \frac{\left\lVert A\right\rVert\cdot\left\lVert P^{\frac{1}{2}}D_{\eta}^{\frac{1}{2}}\right\rVert\cdot\left\lVert Q^{\frac{1}{2}}D_{\gamma}^{\frac{1}{2}}\right\rVert}{2}
		\left(  \left\lVert x^t\right\rVert_{P^{-1}D_{\eta}^{-1}}^2+\left\lVert y^t\right\rVert_{Q^{-1}D_{\gamma}^{-1}}^2\right)
	\end{align*}
	where the first inequality is the Cauchy-Swartz inequality, the second inequality follows by definition of the matrix norm $\lVert A \rVert = \max_w \frac{\lVert Aw\rVert}{\lVert w\rVert}$, the third inequality follows by Lemma \ref{lem:norm}, and the final equality follows since  $ab=\frac{a^2+b^2-(a-b)^2}{2}\leq \frac{a^2+b^2}{2}$. 
	
	Combining the two expressions and re-arranging terms yields
	\begin{align*}
	\left\lVert x^t\right\rVert^2_{P^{-1}D_{{\eta}}^{-1}} + \left\lVert y^t\right\rVert^2_{Q^{-1}D_{\gamma}^{-1}}
	&\leq \frac{\left\lVert x^0\right\rVert^2_{P^{-1}D_{{\eta}}^{-1}} + \left\lVert y^0\right\rVert^2_{Q^{-1}D_{\gamma}^{-1}}+\left\langle x^0, Ay^0\right\rangle}
	{1-\frac{\left\lVert A\right\rVert\cdot\left\lVert P^{\frac{1}{2}}D_{\eta}^{\frac{1}{2}}\right\rVert\cdot\left\lVert Q^{\frac{1}{2}}D_{\gamma}^{\frac{1}{2}}\right\rVert}{2}}.
	\end{align*}
	Note, that the denominator is positive since $\left\lVert P^{\frac{1}{2}}D_{\eta}^{\frac{1}{2}}\right\rVert\cdot\left\lVert Q^{\frac{1}{2}}D_{\gamma}^{\frac{1}{2}}\right\rVert< \frac{2}{||A||}$ and the direction of the inequality was maintained while rearranging terms.
	Thus, the updates are bounded.
	We remark it is also straightforward to bound $||x^t||$ in the standard euclidean space since, by Lemma \ref{lem:norm}, $\lVert x^t\rVert \leq \left\lVert P^{\frac{1}{2}}D_{\eta}^{\frac{1}{2}}\right\rVert \cdot \left\lVert x^t\right\rVert_{P^{-1}D_{{\eta}}^{-1}}\leq \left\lVert P^{\frac{1}{2}}D_{\eta}^{\frac{1}{2}}\right\rVert \cdot \sqrt{\left\lVert x^t\right\rVert^2_{P^{-1}D_{{\eta}}^{-1}} + \left\lVert y^t\right\rVert^2_{Q^{-1}D_{\gamma}^{-1}}}$.
	\end{proof}

In addition to being necessary for the proof of recurrence, Theorem \ref{thm:BoundedOrbits} also allows us to refine our results related to regret from Section \ref{sec:2AgentRegret}. 
Recall that the statement of Theorem \ref{thm:regret} only claims that agent 1's regret is bounded after agent 1 updates and that Proposition \ref{prop:BadRegret} shows that is possible for agent 1 to have large regret after agent 2 updates. 
With Theorem \ref{thm:BoundedOrbits}, we can show that agent 1 will always have bounded regret, regardless of which agent updates last. 

\begin{corollary}\label{cor:ZSRegret}
	Suppose $P$ and $Q$ commute with $D_{\eta}$ and $D_{\gamma}$ respectively and $\left\lVert P^{\frac{1}{2}}D_{\eta}^{\frac{1}{2}}\right\rVert\cdot\left\lVert Q^{\frac{1}{2}}D_{\gamma}^{\frac{1}{2}}\right\rVert< \frac{2}{||A||}$ when updating both agents' strategies with Algorithm \ref{alg:2Agent} in (\ref{eqn:posneg}). 
	Agent 1's regret is bounded when regret is computed after agent 2 updates. 
\end{corollary}

\begin{proof}
	From Proposition \ref{prop:BadRegret}, agent 1's regret is $\sum_{t=0}^T\left\langle 2x-x^{t+1}-x^t,  Ay^t\right\rangle+\left\langle x-x^{T+1},  Ay^{T+1}\right\rangle$.
	The first term is bounded by Theorem \ref{thm:regret}.  
	Moreover, the second term is also bounded since $x^{T+1}$ and $y^{T+1}$ are bounded (Theorem \ref{thm:BoundedOrbits}).  
	Thus, agent 1's regret is bounded -- even when agent 2 updates last. 
\end{proof}

\subsection{The Bound $\left\lVert P^{\frac{1}{2}}D_{\eta}^{\frac{1}{2}}\right\rVert\cdot\left\lVert Q^{\frac{1}{2}}D_{\gamma}^{\frac{1}{2}}\right\rVert< \frac{2}{||A||}$ is Tight}

All three main results in this section require learning rates to be sufficiently small.  
In the following proposition, we show that that the bound of $2/||A||$ on learning rates is tight. 

\begin{proposition}\label{prop:unbounded}
	If the learning rates are too large when both agents use Algorithm \ref{alg:2Agent} in (\ref{eqn:posneg}), then the strategies may diverge -- even if $\left\lVert P^{\frac{1}{2}}D_\eta^{\frac{1}{2}}\right\rVert \cdot \left\lVert Q^{\frac{1}{2}}D_\gamma^{\frac{1}{2}}\right\rVert= \frac{2}{\lVert A \rVert}$. 
\end{proposition}

\begin{figure}[H]
	\centering
	\includegraphics[scale=0.4]{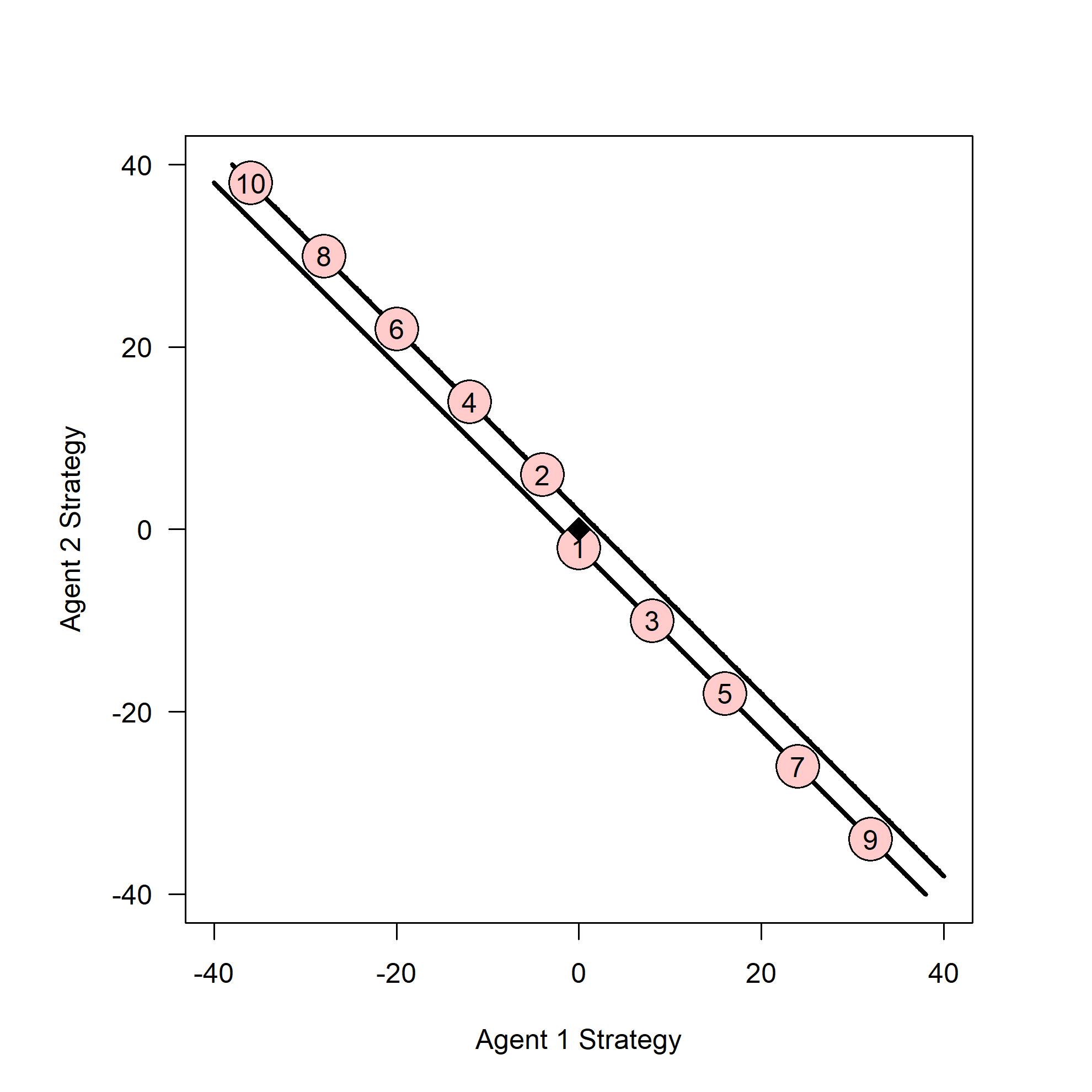}
	\caption{Algorithm \ref{alg:2Agent} applied to $A=[1]$, $P=Q=[1]$, $(x^0,y^0)=(0,-2)$ and $\eta=\gamma=2$.  Since $\left\lVert P^{\frac{1}{2}}D_\eta^{\frac{1}{2}}\right\rVert \cdot \left\lVert Q^{\frac{1}{2}}D_\gamma^{\frac{1}{2}}\right\rVert=  \frac{2}{\lVert A \rVert}$ the level sets of the energy function from Theorem \ref{thm:energy2} are not compact and the strategies diverge. }\label{fig:Counter}
\end{figure}

\begin{proof}
	Let $A=[1]$, $P=Q=[1]$, $(x^0,y^0)=(0,-2)$ and $\eta=\gamma=2$.
	Since $2=\sqrt{\eta\cdot\gamma}=\left\lVert P^{\frac{1}{2}}D_\eta^{\frac{1}{2}}\right\rVert\cdot\left\lVert Q^{\frac{1}{2}}D_\gamma^{\frac{1}{2}}\right\rVert= \frac{2}{\lVert A \rVert}= 2$, Theorem \ref{thm:converge} does not imply and we cannot immediately claim the strategies will remain bounded.
	Using induction, we will show $(x^t,y^t)=\left((-1)^{t}\cdot 4t, (-1)^{t+1}\cdot (4t+2)\right)$.  
	The result trivially holds for $t=0$.

	By the inductive hypothesis, $(x^{t-1},y^{t-1})=\left((-1)^{t-1}\cdot 4(t-1), (-1)^{t}\cdot (4t-2)\right)$. 
	Therefore
	\begin{align*}
	x^t = x^{t-1} + 2 \cdot A y^{t-1} &= (-1)^{t-1} (4t-4) + 2 (-1)^t(4t-2)\\
		&= (-1)^{t}\left(4-4t+8t-4  \right)=(-1)^{t}\cdot 4t.
	\end{align*}
	Similarly for agent 2, 
	\begin{align*}
	y^t = y^{t-1} - 2 \cdot A x^{t} &= (-1)^{t}\cdot (4t-2) - 2(-1)^{t}\cdot 4t\\
	&= (-1)^{t+1}\left(2-4t+8t  \right)=(-1)^{t+1}\cdot (4t+2).
	\end{align*}
	Thus, agent 1's strategy over time is the diverging sequence $\{(-1)^t\cdot 4t\}_{t=0}^\infty$. 
\end{proof}

\subsection{$O\left(1/T\right)$ Time-Average Convergence to Nash in Positive-Negative Definite Games} \label{sec:converge}

\begin{figure}[H]
	\centering
	\includegraphics[scale=0.4]{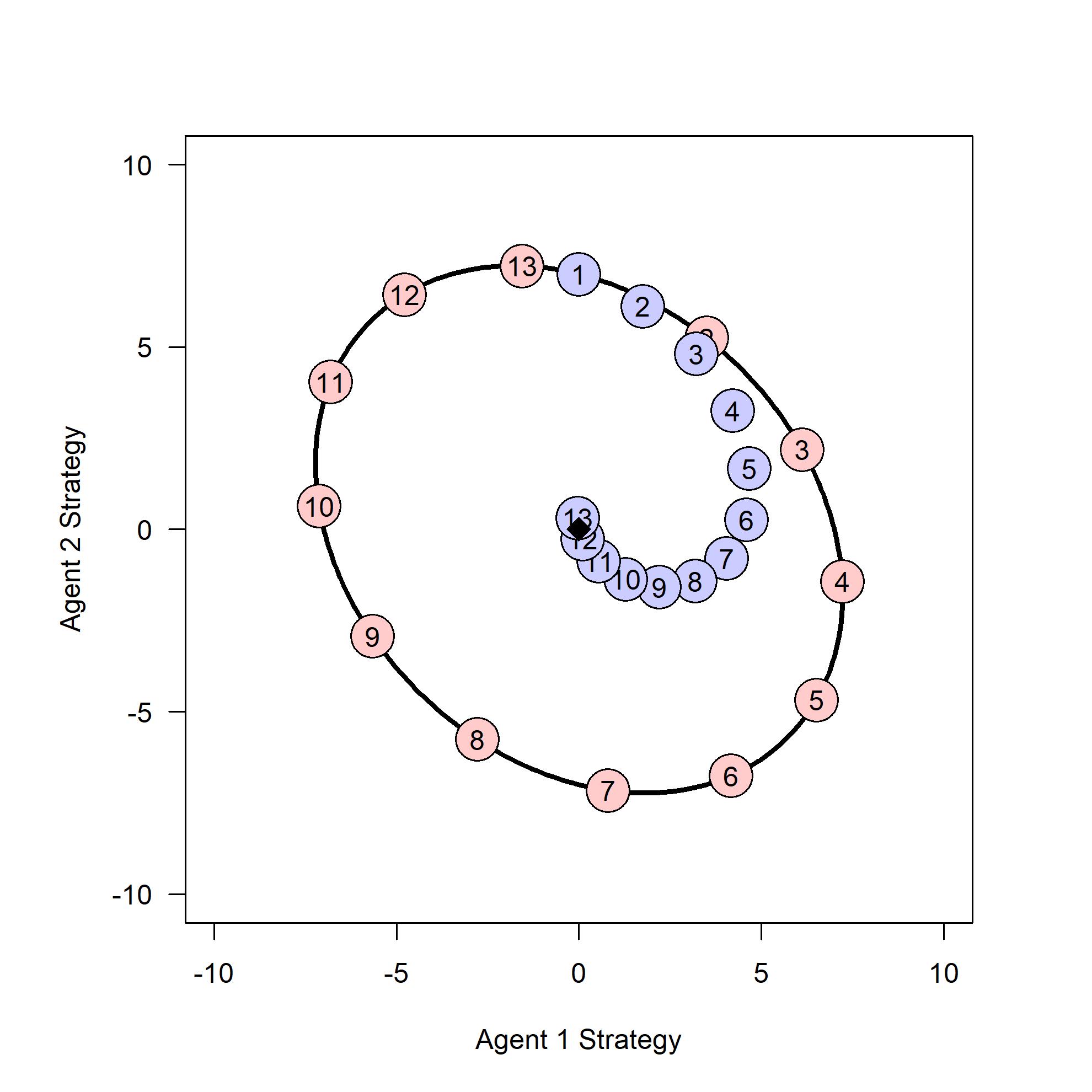}
	\caption{ The time-average strategies (blue) converge to the Nash equilibrium as the strategies (red) cycle around the Nash equilibrium.  }\label{fig:TimeAverage}
\end{figure}

In this section, we show that the time-average of the strategies converge to the set of Nash equilibria at rate $O(1/T)$ as depicted in Figure \ref{fig:TimeAverage}.  
We measure the distance to the set of Nash equilibria by $\left\lVert\sum_{t=0}^{T-1} \frac{PAy^t}{T}\right\rVert$.  
This is a standard measure since $y^*$ is a Nash equilibrium if and only if $PAy^*=\vec{0}$. 

\begin{theorem}\label{thm:converge}
	Suppose $P$ and $Q$ commute with $D_{\eta}$ and $D_{\gamma}$ respectively and $\left\lVert P^{\frac{1}{2}}D_{\eta}^{\frac{1}{2}}\right\rVert\cdot\left\lVert Q^{\frac{1}{2}}D_{\gamma}^{\frac{1}{2}}\right\rVert< \frac{2}{||A||}$ when updating both agents' strategies with Algorithm \ref{alg:2Agent} in (\ref{eqn:posneg}).
	Then agent 2's strategy has $O\left(1/T\right)$ time-average convergence to the set of Nash equilibria.  
	Formally, there exists a constant $c$ such that for all $i$, $\left\lVert\sum_{t=0}^{T-1} \frac{PAy^t}{T}\right\rVert \leq c/T $.
	Symmetrically, agent 1's strategy also has $O\left(1/T\right)$ time-average convergence to the set of Nash equilibria. 
\end{theorem}

Perhaps surprisingly, we do not use the regret property to prove time-average convergence.  
Rather, time-average convergence follows immediately from the compact level sets of the energy function (Theorem \ref{thm:BoundedOrbits}). 

\begin{proof}[Proof of Theorem \ref{thm:converge}]
	By Theorem \ref{thm:BoundedOrbits}, $\{x^t\}_{t=0}^\infty$ belongs to a compact set and there exists a $c$ such that $\lVert x^t-x^0\rVert\leq \lVert x^t\rVert +\lVert x^0\rVert \leq c$ for each iteration $t$. 
	Recall that $x^T=x^{T-1}+PAy^{T-1}=x^0+\sum_{t=0}^{T-1}PAy^{t}$. 
	Thus, $\left\lVert\sum_{t=0}^{T-1}PAy^{t}\right\rVert =  \lVert x^T-x^0\rVert \leq c$ completing the claim for agent 2. 
	
	The result for agent 1 follows identically using $y^T-y^0=\sum_{t=1}^T-QA^\intercal x^t$. 
\end{proof}

We remark that the constant $c$ can be computed directly using the bound in Theorem \ref{thm:BoundedOrbits}.
Once again, the bound on the learning rates is tight.

\begin{proposition}\label{prop:noconverge}
If the learning rates are too large when both agents use Algorithm \ref{alg:2Agent} in (\ref{eqn:posneg}), then time-average of the strategies may fail to converge -- even if $\left\lVert P^{\frac{1}{2}}D_\eta^{\frac{1}{2}}\right\rVert \cdot \left\lVert Q^{\frac{1}{2}}D_\gamma^{\frac{1}{2}}\right\rVert= \frac{2}{\lVert A \rVert}$. 
\end{proposition}

\begin{proof}
	In Proposition \ref{prop:unbounded}, we showed that for $A=[1]$, $P=Q=1$, $(x^0,y^0)=(0,-2)$ and $\eta=\gamma=2$ that $x^t=(-1)^{t}\cdot 4t$. 
	Therefore, agent 1's time-average strategy alternates between $\sum_{t=0}^{2T} \frac{(-1)^t4\cdot t}{2T}=2$ and $\sum_{t=0}^{2T+1} \frac{(-1)^t4\cdot t}{2T+1}=\frac{-4T}{2T+1}\to-2$ on even and odd iterations thereby completing the proof. 
\end{proof}

\section{An Algorithm for Multiagent Systems}
\label{sec:Multi}

In this section, we extend our previous results to the multiagent system. 

\begin{align}
	\max_{x_i\in{\cal X}_i} \left\langle x_i, \sum_{j\neq i} A^{(ij)}x_j\right\rangle \ for \ all \ i \tag{Network Game} \label{eqn:MultiAgentGame}
\end{align}

Perhaps the most natural way to extend alternating gradient descent is to have agents iteratively take turns in a round-robin, i.e., agent 1 updates, then agent 2, and so on. 
However, in Section \ref{sec:Round}, we show this idea fails miserably -- regret can grow linearly. 
The secret to the success of alternating gradient descent (Algorithm \ref{alg:2Agent}), doesn't actually have anything to do with the perceived fairness of having agents take turns.

Instead, in Section \ref{sec:physics} we extend the results for alternating gradient descent by understanding it as an approximation of a Hamiltonian system, a well-understood physical system. 
Specifically, alternating gradient descent naturally arises when approximating this continuous-time system using a symplectic integrator -- specifically Verlet integration.
In Section \ref{sec:Reduce}, we reduce the multiagent game to a 2-agent game through the use of two meta-agents and in Section \ref{sec:MultiRegret}, we extend the regret and conservation guarantees of Sections \ref{sec:2AgentRegret}-\ref{sec:ConservationVolume} to the multiagent case. 
Finally, in Section \ref{sec:ZeroMulti}, we provide time-average convergence guarantees for positive-negative definite multiagent games. 

\subsection{Gradient Descent in a Round Robin}\label{sec:Round}

First, we consider a ``fair'' implementation of gradient descent where agents take turns updating and show that the algorithm can have linear regret. 

\begin{varalgorithm}{RoundGD}
	\caption{Multiagent Gradient Descent with Agents Taking Turns}\label{alg:Round}
	\label{alg:euclid}
	\begin{algorithmic}[1]
		\Procedure{RoundGD}{$A,x^0,\eta$}\Comment{Payoff Matrices, Initial Strategies and Learning Rates}
		\For{\texttt{$t=1,...,T$}}
			\For{\texttt{$i=1,...,N$}}
				\State $x_i^t:= x_i^{t-1} + D_{\eta_i} \sum_{j< i} A^{(ij)}x_j^{t} + D_{\eta_i} \sum_{j> i} A^{(ij)}x_j^{t-1}$ \Comment{Agent Updates Strategies \underline{After} Seeing Updated Strategies of All Agents Updating Prior in the Round Robin}\label{line:Round}
			\EndFor
		\EndFor
		\EndProcedure
	\end{algorithmic}
\end{varalgorithm}

%\begin{algorithm}[H]
%	\SetAlgoLined
%	%	\KwResult{Write here the result }
%	\textbf{Input:} \\
%	\hspace*{1.2em} Agent $i's$ payoff matrix against agent $j$: $A^{(ij)}$ for $i=1,...,n$ and $j\neq i$\\
%	\hspace*{1.2em} Agent $i$ initial strategy: $x_i^0\in {\cal X}_i$ for $i=1,...,n$\\
%	\hspace*{1.2em} Agent $i$ learning rate: ${\eta}_i>0$ for $i=1,...,n$\\
%
%	
%	\For{$t=1,...,T$}{
%		\For{$i=1,...,N$}{
%		$x_i^t:= x_i^{t-1} + D_{\eta_i} \sum_{j< i} A^{(ij)}x_j^{t} + D_{\eta_i} \sum_{j> i} A^{(ij)}x_j^{t-1}$ \tcp*{Agent Updates Strategies \underline{After} Seeing Updated Strategies of All Agents Updating Prior in the Round Robin}\label{line:Round}
%		}
%	}
%	\caption{Gradient Descent Applied in a Round Robin.}\label{alg:Round}
%\end{algorithm}

\begin{remark}
	If line \ref{line:Round} of Algorithm \ref{alg:Round} is replaced with $x_i^t:= x_i^{t-1} + D_{\eta_i} \sum_{j\neq i} A^{(ij)}x_j^{t-1}$ then Algorithm \ref{alg:Round} reduces to the standard implementation of gradient descent with simultaneous updates. 
\end{remark}

\begin{proposition}
	If agents take turns using gradient descent in a multiagent setting (Algorithm \ref{alg:Round}), then an agent's regret can grow linearly. \label{prop:robin}
\end{proposition}

\begin{proof}
	Consider the simple 2-agent zero-sum game with $A^{(12)}=[1]$ and $A^{(21)}=[-1]$ with initial strategies $x^0_1=x^0_2=1$. Suppose both agents update according to Algorithm \ref{alg:2Agent} with learning rate $\eta_1=\eta_2=1$. 
	The agents strategies will cycle every 6 iterations (12 updates) as shown in Figure \ref{fig:cycle}.
	As such, will gain 0 utility from any 6 consecutive iterations -- $+6$ from when agent 1 updates and $-6$ from when agent 2 updates. 
	
	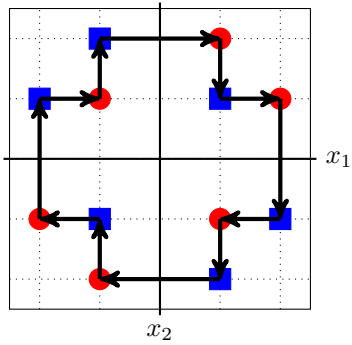
\begin{figure}[H]
		\centering
	\begin{tikzpicture}[scale=.8]
	\draw[draw=black] (-2.5,-2.5) rectangle (2.5,2.5);
	\foreach \i in {-2,-1,0,1,2}
	{
		\draw[dotted] (-2.5,\i)--(2.5,\i);	
		\draw[dotted] (\i,-2.5)--(\i,2.5);
	}
	\foreach \i in {0}
	{
		\draw[thick] (-2.6,\i)--(2.6,\i);	
		\draw[thick] (\i,-2.6)--(\i,2.6);
	}
	\node[right] at (2.6,0) {$x_1$};
	\node[below] at (0,-2.6) {$x_2$};
	\coordinate (A1) at (2,1);
	\coordinate (B1) at (2,-1);
	\coordinate (A2) at (1,-1);
	\coordinate (B2) at (1,-2);
	\coordinate (A3) at (-1,-2);
	\coordinate (B3) at (-1,-1);
	\coordinate (A4) at (-2,-1);
	\coordinate (B4) at (-2,1);
	\coordinate (A5) at (-1,1);
	\coordinate (B5) at (-1,2);
	\coordinate (A6) at (1,2);
	\coordinate (B6) at (1,1);
	
	\node[mark size=4pt,color=red] at (A1) {\pgfuseplotmark{*}};
	\node[mark size=4pt,color=blue] at (B1){\pgfuseplotmark{square*}};
	\node[mark size=4pt,color=red] at (A2) {\pgfuseplotmark{*}};
	\node[mark size=4pt,color=blue] at (B2){\pgfuseplotmark{square*}};
	\node[mark size=4pt,color=red] at (A3) {\pgfuseplotmark{*}};
	\node[mark size=4pt,color=blue] at (B3){\pgfuseplotmark{square*}};
	\node[mark size=4pt,color=red] at (A4) {\pgfuseplotmark{*}};
	\node[mark size=4pt,color=blue] at (B4){\pgfuseplotmark{square*}};
	\node[mark size=4pt,color=red] at (A5) {\pgfuseplotmark{*}};
	\node[mark size=4pt,color=blue] at (B5){\pgfuseplotmark{square*}};
	\node[mark size=4pt,color=red] at (A6) {\pgfuseplotmark{*}};
	\node[mark size=4pt,color=blue] at (B6){\pgfuseplotmark{square*}};

	\draw[,>=stealth',->,ultra thick] (A1)->(B1);
	\draw[,>=stealth',->,ultra thick] (B1)->(A2);
	\draw[,>=stealth',->,ultra thick] (A2)->(B2);
	\draw[,>=stealth',->,ultra thick] (B2)->(A3);
	\draw[,>=stealth',->,ultra thick] (A3)->(B3);
	\draw[,>=stealth',->,ultra thick] (B3)->(A4);
	\draw[,>=stealth',->,ultra thick] (A4)->(B4);
	\draw[,>=stealth',->,ultra thick] (B4)->(A5);
	\draw[,>=stealth',->,ultra thick] (A5)->(B5);
	\draw[,>=stealth',->,ultra thick] (B5)->(A6);
	\draw[,>=stealth',->,ultra thick] (A6)->(B6);
	\draw[,>=stealth',->,ultra thick] (B6)->(A1);
	\end{tikzpicture}
	\caption{The red circles correspond to agents' strategies after agent 1 updates and the blue squares correspond to the strategies after agent 2 updates. Agent 1's total utility from the red circles is $(1\cdot 2) + (2\cdot 1) +(1\cdot -1)+(-1\cdot-2)+(-2\cdot-1)+(-1\cdot 1)=6$ and agent 1's total utility from the blue squares is $-6$. }\label{fig:cycle}
	\end{figure}

	Now consider the addition of $k$-dummy agents where $A^{(ij)}=[0]$ for all $i=3,...,k+2$ and $j=1,...,k+2$.  
	Further suppose that all agents update according to Algorithm \ref{alg:Round}, i.e., agent 1 updates, then agent 2, and so on. 
	When agent $j>2$ updates, no agents will change their strategies since their payoff matrices are all zero. 
	As a result, in the space of the first two agents, for each update spent at a red circle in Figure \ref{fig:cycle} there will be $k+1$ updates at the following blue square.  
	Therefore every 6 round-robins (every cycle) will contribute $-6k$ to agent 1's cumulative utility implying agent 1's regret with respect to $x_1=0$ is $\Theta(k\cdot T)$. 
\end{proof}

\subsection{Designing a Multiagent Algorithm Based on Physics}\label{sec:physics}

As demonstrated by Proposition \ref{prop:robin}, the reason Algorithm \ref{alg:2Agent} works isn't because it makes agents take turns in a seemingly fair fashion.
Rather, Algorithm \ref{alg:2Agent} works because it is the result of a deep understanding of the physical system that drives gradient descent. 
By \cite{Bailey19Hamiltonian}, the continuous-time version of 2-agent gradient descent is a Hamiltonian system (e.g., Earth-moon system) where agent 1 corresponds to ``position'' and agent 2 corresponds to ``momentum''.
The continuous-time variant has nice optimality and stability guarantees: $O(1/T)$ time-average regret and recurrence in zero-sum games \cite{Mertikopoulos2018CyclesAdverserial}.
Algorithm \ref{alg:2Agent} is obtained by applying Verlet integration \cite{Bailey2019Regret}, an integration technique well-suited for approximating Hamiltonian dynamics \cite{Hairer2006EnergyConserve}, to the underlying Hamiltonian system.

Verlet integration corresponds to simply alternating between updating ``position'' and ``momentum'' in the underlying system.
However, in the multiagent case, it is unclear which agents correspond to position and momentum respectively. 
Like \cite{Bailey19Hamiltonian}, which shows the continuous-time system is Hamiltonian, we resolve this issue by allowing agents to be both position and momentum.
However, instead of double-counting each agent as in \cite{Bailey19Hamiltonian}, we instead duplicate each agent and build a game between the original and duplicated agents. Specifically, we allow the original agent $i$ control the strategy $x_i\in {\cal X}_i$ while their doppelganger controls strategy $y_i\in {\cal X}_i$ resulting in the following game. 
\begin{equation}\label{eqn:NetworkGame}
\begin{aligned}
\max_{x_i\in {\cal X}_i} \left \langle x_i, \sum_{j\neq i} A^{(ij)} y_j\right \rangle \ for \ all \ i\\
\max_{y_i\in {\cal X}_i} \left \langle y_i, \sum_{j\neq i} A^{(ij)} x_j \right \rangle\ for \ all \ i\\
\end{aligned}\tag{Network Game with Duplicated Agents}
\end{equation}

\begin{varalgorithm}{AltGD}
	\caption{Verlet Integration of Continuous Gradient Descent with Duplicated Agents (Alternating Gradient Descent for Multiagent Systems)}\label{alg:Verlet}
	\label{alg:euclid}
	\begin{algorithmic}[1]
		\Procedure{AltGD}{$A,x^0,y^0,\eta,\gamma$}\Comment{Payoff Matrices, Initial Strategies and Learning Rates}
%		\State \Comment{Duplicated Agents May Use Different Initial Strategies and Learning Rates.}
%		\State \Comment{The original agent $i$ has the objective $\max_{x_i\in {\cal X}_i} \left\langle x_i, \sum_{j\neq i} A^{(ij)}y_j\right\rangle$.}
%		\State \Comment{The duplicated agent $i$ has the objective $\max_{y_i\in {\cal X}_i} \left\langle y_i, \sum_{j\neq i} A^{(ij)}x_j\right\rangle$.}
		\For{\texttt{$t=1,...,T$}}
			\For{\texttt{$i=1,...,N$}}
				\State $x_i^t:= x_i^{t-1} + D_{\eta_i} \sum_{j\neq i} A^{(ij)}y_j^{t-1}$ \Comment{Update Based on Previous Iteration}\label{line:original}
			\EndFor
			\For{\texttt{$i=1,...,N$}}
				\State $y_i^t:= y_i^{t-1} + D_{\gamma_i} \sum_{j\neq i} A^{(ij)}x_j^{{{t}}}$ \Comment{Update Based on Current Iteration.}\label{line:duplicate}
			\EndFor
		\EndFor
		\EndProcedure
	\end{algorithmic}
\end{varalgorithm}

%\begin{algorithm}[H]
%	\SetAlgoLined
%%	\KwResult{Write here the result }
%	\textbf{Input:} \\
%	\hspace*{1.2em} Agent $i's$ payoff matrix against agent $j$: $A^{(ij)}$ for $i=1,...,n$ and $j\neq i$\\
%	\hspace*{1.2em} Agent $i$ initial strategy: $x_i^0\in {\cal X}_i$ for $i=1,...,n$\\
%	\hspace*{1.2em} Agent $i$ learning rate: ${\eta}_i>0$ for $i=1,...,n$\\
%	\hspace*{1.2em} Duplicate Agent $i$ initial strategy: $y_i^0\in {\cal X}_i$ for $i=1,...,n$\\
%	\hspace*{1.2em} Duplicate Agent $i$ learning rate: ${\gamma}_i>0$ for $i=1,...,n$\\
%	
%	\tcc{The original agent $i$ has the objective $\max_{x_i\in {\cal X}_i} \left\langle x_i, \sum_{j\neq i} A^{(ij)}y_j\right\rangle$}
%	\tcc{The duplicated agent $i$ has the objective $\max_{y_i\in {\cal X}_i} \left\langle y_i, \sum_{j\neq i} A^{(ij)}x_j\right\rangle$}
%	
%	\For{$t=1,...,T$}{
%		$x_i^t:= x_i^{t-1} + D_{\eta_i} \sum_{j\neq i} A^{(ij)}y_j^{t-1}$ \tcp*{Original Agent Updates Strategies via Gradient Descent}\label{line:original}
%		
%		$y_i^t:= y_i^{t-1} + D_{\gamma_i} \sum_{j\neq i} A^{(ij)}x_j^{{{t}}}$ \tcp*{Duplicate Agent Updates Strategies via Gradient Descent \underline{After} Seeing Original Agents' Updated Strategies.}\label{line:duplicate}
%	}
%	\caption{Verlet Integration of Continuous Gradient Descent with Duplicated Agents (Alternating Gradient Descent for Multiagent Systems)}\label{alg:Verlet}
%\end{algorithm}

\begin{remark}
	If the strategies and learning rates are initialized so that $x_i^0=y_i^0$ and ${\eta}_i={\gamma}_i$, and if line \ref{line:duplicate} of Algorithm \ref{alg:Verlet} is replaced with $y_i^t:= y_i^{t-1}+ {\eta}_i \sum_{j\neq i} A^{(ij)} x_j^{t-1}$ (the original and duplicate agents update simultaneously) then $x_i^t=y_i^t$ in each iteration and the algorithm reduces to the standard version of gradient descent where the original agents simultaneously update with respect to the original agents' previously played strategies (Algorithm \ref{alg:GradientMulti}). 
\end{remark}

Since no duplicate agent actually exists in many economic settings, Algorithm \ref{alg:Verlet} should primarily be used when interested in aggregate behavior, i.e., $\bar{x}_i^T=\sum_{t=0}^T x_i^t/T$. 
Such applications are fairly standard in GANs and in other simulated environments such as bargaining and resource allocation problems that seeks a Nash equilibrium without agents directly sharing their payoff matrices, e.g., \cite{Shahrampour20OnlineAllocation}. 

\subsection{Reducing the Multiagent System to a 2-Agent Game}\label{sec:Reduce}

	By introducing two meta-agents to control the original and duplicated agents, it is possible to express (\ref{eqn:MultiAgentGame}) and Algorithm \ref{alg:Verlet} as  \ref{eqn:2AgentGame} and Algorithm \ref{alg:2Agent} respectively. 
	Formally, we consider the following meta-game:

		\begin{equation}\label{eqn:MetaGame}
	\begin{aligned}
	\max_{\bar{x}\in \times_i{\cal X}_i} \left \langle \bar{x}, \bar{A}\bar{y}\right \rangle \\
	\max_{\bar{y}\in \times_i{\cal X}_i} \left \langle \bar{y}, \bar{A}\bar{x}\right \rangle
	\end{aligned}\tag{Meta-Game}
	\end{equation}
	where $\bar{x}=[x_1,x_2,...,x_n]$, $\bar{y}=[y_1,y_2,...,y_n]$ and 
	\begin{align*}
		\bar{A}=\left[\begin{array}{ c c c c c}
			A^{(11)}=0 & A^{(12)} & A^{(13)} & \cdots & A^{(1n)}\\
			A^{(21)} & A^{(22)}=0 & A^{(23)} & \cdots & A^{(2n)}\\
			A^{(31)} & A^{(32)} & A^{(33)}=0 & \cdots & A^{(3n)}\\
			\vdots & \vdots & \vdots & \ddots & \vdots\\
			A^{(n1)} & A^{(n2)} & A^{(n3)} & \cdots & A^{(nn)}=0\\
		\end{array}\right]
	\end{align*}
	
	In Theorem \ref{thm:equiv}, we show the sets of Nash equilibria for (\ref{eqn:MetaGame}) and (\ref{eqn:MultiAgentGame}) are equivalent. 
	Moreover, in Theorem \ref{thm:same}, we show that the Algorithms \ref{alg:2Agent} and \ref{alg:Verlet} result in the same updates for their respective games. 
	As such, most of our results for 2-agent systems readily extend to the multiagent setting.

	This reduction emphasizes the importance of allowing Algorithm \ref{alg:2Agent} to run with an arbitrary vector of learning rates.  
	In the multiagent system, agents individually select their learning rates and therefore do not necessarily use the same learning rates.  
	When applying Algorithm \ref{alg:2Agent} to (\ref{eqn:MetaGame}), the meta-agents will have different learning rates associated with each agent. 
	Notably, our generalization still allows an individual to use different learning rates for different strategies, even in the multiagent setting. 
	However, just like  the 2-agent setting, we see no algorithmic benefit for a single agent to use a vector of learning rates.
	
	\begin{theorem}\label{thm:equiv}
		The strategies $\left(\bar{x}^*, \bar{y}^*\right)$ are a Nash equilibrium for (\ref{eqn:MetaGame})
		if and only if $\bar{x}^*=\left[ \bar{x}_1^*,\cdots \bar{x}_n^*\right]$ and $\bar{y}^*=\left[ \bar{y}_1^*,\cdots \bar{y}_n^*\right]$ are both Nash equilibria (possibly the same) for (\ref{eqn:MultiAgentGame}). 
	\end{theorem}

	\begin{proof}
		First, let $\left(\bar{x}^*, \bar{y}^*\right)$ be a Nash equilibrium of (\ref{eqn:MetaGame}). 
		Then $\bar{A}\bar{y}^*=\vec{0}$;  otherwise, meta-agent 1 could increase their utility by $\lVert \bar{A}\bar{y}^* \rVert^2>0$ by updating their strategy to $\bar{x}^*+ \bar{A}\bar{y}^*$. 
		Therefore, by definition of $\bar{A}$ and $\bar{y}$, $\sum_{j\neq i} A^{(ij)}{\bar{y}_j^*}=\vec{0}$ for each $i$. 
		It then holds that $\bar{y}_i^*$ is a best response to $\bar{y}^*_{-i}$ (all strategies but agent $i$) in (\ref{eqn:MultiAgentGame}) since 
		$y_i\sum_{j\neq i} A^{(ij)}{\bar{y}_j^*}=\vec{0}$ for all $y_i\in {\cal X}_i$.  
		This holds for each agent $i$ and therefore $y^*$ i a Nash equilibrium for (\ref{eqn:MultiAgentGame}).
		The argument for $\bar{x}^*$ follow identically. 
		
		Next, let $y^*$ be a Nash equilibrium of (\ref{eqn:MultiAgentGame}).  Then $\sum_{j\neq i} A^{(ij)}{y_j^*}=\vec{0}$ for each $i$ since otherwise agent $i$ could increase their utility by $||\sum_{j\neq i} A^{(ij)}{y_j^*}||^2$ with the strategy $y_i^*+ \sum_{j\neq i} A^{(ij)}{y_j^*}$.  As such, $\bar{A}y^*=\vec{0}$ and $y^*$ is a Nash equilibrium of (\ref{eqn:MetaGame}). 
		The argument holds identically for ${x}^*$. 
	\end{proof}

\begin{theorem}\label{thm:same}
	Suppose $\{x^t,y^t\}_{t=0}^T$ is obtained by updating (\ref{eqn:MultiAgentGame}) with Algorithm \ref{alg:Verlet} using learning rates $\eta=(\eta_1,\eta_2,...,\eta_N)$ and $\gamma= (\gamma_1, \gamma_2, ..., \gamma_N)$ and initial strategies $(x^0,y^0)$. 
	Further, suppose $\{\bar{x}^t, \bar{y}^t\}_{t=0}^T$ is obtained by updating   (\ref{eqn:MetaGame}) with Algorithm \ref{alg:2Agent} with learning rates $\eta$ and $\gamma$ and initial strategy $(\bar{x}^0, \bar{y}^0)=(x^0, y^0)$. 
	Then $(\bar{x}^t, \bar{y}^t)=(x^t,y^t)$ for all $t=0,...,T$. 
\end{theorem}

Theorem \ref{thm:same} holds trivially by induction since the optimization problem is separable with respect to each agent.

\subsection{Regret and Conservation in Multiagent Games}\label{sec:MultiRegret}

We show that Theorems \ref{thm:regret} (regret) and \ref{thm:Volume} (volume conservation) both extend to this setting. 
We remark that Theorem \ref{thm:Actualization} (self-actualization) also extends, however, as discussed in Section \ref{sec:physics}, Algorithm \ref{alg:Verlet} is best used in settings where only aggregate information ($\sum_{t=1}^T x^t/T$) is of interest. 

\begin{theorem}[$1/T$ Time-Average Regret]\label{thm:MultiRegret}
	If agent $i$ updates their strategies with Algorithm \ref{alg:Verlet} in  (\ref{eqn:NetworkGame}) with an \underline{arbitrary} vector of fixed learning rates ${\eta_i}$, then their time-average regret with respect to an arbitrary fixed strategy $x_i$ in iteration $T$ is $O\left(1/T\right)$, regardless of how their opponents update. 
\end{theorem}

%\begin{theorem}[Invariant Energy]\label{thm:MultiEnergy}
%	Suppose that $\bar{A}$ is invertible and commutes with $D_\eta$ and $D_\gamma$. 
%	Then the perturbed energy $\langle \bar{x}^t_i, \bar{y}^t\rangle+ \langle {\bar x}^t, \bar{A}^{-1}D_{\bar{\eta}}^{-1}\bar{x}^t\rangle + \langle \bar{y}^t, \bar{A}^{-1}D_{\bar{\gamma}}^{-1}\bar{y}^t\rangle$ is invariant when agents update their strategies with Algorithm \ref{alg:Verlet}.  
%\end{theorem}

%\begin{remark}
%In \cite{Bailey20Uniqueness}, it was shown that almost every general-sum game has a unique Nash equilibrium and that ${\bar{A}}^{-1}$ exists so long as no agent controls more than half the strategies -- notably this is equivalent to a square payoff matrix in 2-agent games.  The same result also holds for coordination games.  However, zero-sum games additionally require that there are an even number of strategies.  We provide a different energy function for zero-sum and coordination games in Section \ref{sec:ZeroMulti} that places no restrictions on $\bar{A}$.
%\end{remark}

\begin{theorem}[Volume Conservation]\label{thm:MultiVolume}
		Algorithm \ref{alg:2Agent} in (\ref{eqn:NetworkGame}) is volume preserving for any measurable set of initial conditions. 
\end{theorem}

Theorem \ref{thm:MultiVolume} follow immediately by Theorem \ref{thm:Volume} after reducing (\ref{eqn:MultiAgentGame}) to (\ref{eqn:MetaGame}). 
Theorem \ref{thm:MultiRegret} almost follows similarly; certainly if all agents use Algorithm \ref{alg:Verlet}, then the corresponding meta-agent has $O(1/T)$ time-average regret. 
Moreover, since $\langle \bar{x}, \bar{A}\bar{y}\rangle$ is separable with respect to each individual agent's strategy $x_i$, each agent also obtains $O(1/T)$ time-average regret. 
However, Theorem \ref{thm:MultiRegret} only requires that agent 1 uses the update rule in Algorithm \ref{alg:Verlet}. 
To see that agent 1 still obtains $O(1/T)$ time-average regret regardless of other agents, we consider the meta-game played between $x_i$ and the meta-agent $\bar{y}$ where the meta-agent is using the same updates as in the original (\ref{eqn:MultiAgentGame}).

\begin{proof}[Proof of Theorem \ref{thm:MultiRegret}]
	Consider the following two-agent game:
	
	\begin{equation}\label{eqn:MetaGame3}
	\begin{aligned}
	\max_{x_i\in {\cal X}_i} \left \langle {x}_i, \bar{A}_{i\cdot}\bar{y}\right \rangle \\
	\max_{\bar{y}\in \times_i{\cal X}_i} \left \langle \bar{y}, \bar{A}_{\cdot i}{x}_i\right \rangle
	\end{aligned}\tag{Meta-Game for Agent $i$}
	\end{equation}
	where $\bar{A}_{i\cdot}= \left[ \begin{array}{c c c c}A^{(i1)} & A^{(i2)} & \cdots & A^{(in)}\end{array}\right]$ is the rows of $\bar{A}$ corresponding to agent $i$'s payoff matrices against other agents and where $\bar{A}_{\cdot i}$ is the columns of $\bar{A}$ corresponding to other agents' payoffs against agent $i$. 
	
	Let $\{\hat{x}^t,\hat{y}\}_{t=0}^T$ be the updates obtained in (\ref{eqn:MultiAgentGame}) where the original agent $i$ uses alternating gradient descent and let $\bar{y}^t=\hat{y}^t$ for all $t=0,...,T$.  This selection implies $x_i^t=\hat{x}_i^t$ since $x_i^t$ and $\hat{x}_i^t$ are updated with gradient descent with the same history of opponent play. Thus agent $i$'s utility and regret are the same in both (\ref{eqn:MultiAgentGame}) and (\ref{eqn:MetaGame3}).  
	By Theorem \ref{thm:MultiRegret}, agent $i$ has $O(1/T)$ time-average regret in (\ref{eqn:MetaGame3}) and therefore also has $O(1/T)$ time-average regret in (\ref{eqn:MultiAgentGame}). 
\end{proof}

	\subsection{Multiagent Positive-Negative Definite  Games}\label{sec:ZeroMulti}

	Similar to Section \ref{sec:ZeroSum}, we introduce a \ref{eqn:MultiPosNeg} and show that Algorithm \ref{alg:Verlet} conserves energy and achieves O(1/T) time-average convergence to the set of Nash equilibria. 

	\begin{equation}\label{eqn:MultiPosNeg}
	\begin{aligned}
	\max_{x_i\in {\cal X}_i} \left \langle x_i, P_i\sum_{j\neq i} A^{(ij)} x_j\right \rangle \ for \ all \ i
	\end{aligned}\tag{Network Positive-Negative Definite Game}
	\end{equation}
	where $A^{(ji)}=-[A^{(ij)}]^\intercal$.

Similarly, a \ref{eqn:MultiPosPos} is
\begin{equation}\label{eqn:MultiPosPos}
\begin{aligned}
\max_{x_i\in {\cal X}_i} \left \langle x_i, P_i\sum_{j\neq i} A^{(ij)} x_j\right \rangle \ for \ all \ i
\end{aligned}\tag{Network Positive-Positive Definite Game}
\end{equation}
where $A^{(ji)}=[A^{(ij)}]^\intercal$.  

%\begin{remark}
%	Ideally, we could work with a more general set of positive definite matrices.  
%	Specifically, it would be preferable if our results held if agent $i$'s payoff matrix against agent $j$ was $P^{(ij)}A^{(ij)}$.  
%	Unfortunately, there is not a clear way to express $\sum_{j\neq i} A^{(ij)}y_j^t$ in terms of $x^{t+1}_i$, $x^{t}_i$, $D_{\eta_i}$, and $P^{(ij)}$ as in the proof of Theorem \ref{thm:energy}.  
%\end{remark}
	Let
\begin{align*}
\bar{P}=\bar{Q}=\left[\begin{array}{c c c c} P_1 & 0 &\cdots & 0 \\0 & P_2 & \cdots &0\\0&0&\ddots&0 \\0 & 0&\cdots&P_n \end{array} \right].
\end{align*}
Then the multi-agent network positive-negative definite game can be reduced to a two-agent positive negative definite game with payoff matrices $\bar{P} \bar{A}$ and $-\bar{Q} \bar{A}^\intercal$ ($\bar{Q}\bar{A}^\intercal$ for the positive-positive definite game). 
Thus the invariant energy functions from Section \ref{sec:ZeroSum} immediately extend to the network setting.

\begin{theorem}[Invariant Energy for (\ref{eqn:MultiPosNeg})]\label{thm:MultiEnergy2}
	Suppose $\bar{P}$ and $\bar{Q}$ commute with $D_{\bar{\eta}}$ and $D_{\bar{\gamma}}$ respectively. 
	Then the perturbed energy $\left\lVert {x}^t\right\rVert^2_{\bar{P}^{-1}D_{\bar{\eta}}^{-1}} + \left\lVert {y}^t\right\rVert^2_{\bar{Q}^{-1}D_{\bar{\gamma}}^{-1}} + \langle {x}^t, \bar{A}{y}^t\rangle$ is invariant when agents play (\ref{eqn:MultiPosNeg}) and update their strategies with Algorithm \ref{alg:Verlet}. 
\end{theorem}

Note that if $D_{\eta_i}$ commutes with $P_i$ then $\bar{P}$ commutes with $D_{\bar{\eta}}$ since both matrices are block diagonal and therefore any theorems that require that $D_{\bar{\eta}}$ and $\bar{P}$ commute hold in most standard applications of online optimization.

\begin{theorem}[Invariant Energy for (\ref{eqn:MultiPosNeg})]\label{thm:MultiEnergy3}
	Suppose $\bar{P}$ and $\bar{Q}$ commute with $D_{\bar{\eta}}$ and $D_{\bar{\gamma}}$ respectively.  
	Then the perturbed energy $\left\lVert {x}^t\right\rVert^2_{\bar{P}^{-1}D_{\bar{\eta}}^{-1}} - \left\lVert {y}^t\right\rVert^2_{\bar{Q}^{-1}D_{\bar{\gamma}}^{-1}} + \langle {x}^t, \bar{A}{y}^t\rangle$ is invariant when agents play (\ref{eqn:MultiPosPos}) and update their strategies with Algorithm \ref{alg:Verlet}. 
\end{theorem}

Moreover, following directly from Theorems \ref{thm:Recurrence}, \ref{thm:BoundedOrbits}, and \ref{thm:converge}, Algorithm \ref{alg:Verlet}, is Poincar\'{e} recurrent, has bounded orbits, and converges to the set of Nash equilibria at rate O(1/T) in (\ref{eqn:MultiPosNeg}). 

\begin{theorem}[Recurrence, Bounded Orbits, and Convergence]\label{thm:multiResult}
	Suppose $\bar{P}$ and $\bar{Q}$ commute with $D_{\bar{\eta}}$ and $D_{\bar{\gamma}}$ respectively and $\left\lVert \bar{P}^{\frac{1}{2}}D_{\bar{\eta}}^{\frac{1}{2}}\right\rVert\cdot\left\lVert \bar{Q}^{\frac{1}{2}}D_{\bar{\gamma}}^{\frac{1}{2}}\right\rVert< \frac{2}{||\bar{A}||}$ when updating both agents' strategies with Algorithm \ref{alg:Verlet} in (\ref{eqn:MultiPosNeg}). 
	Then 
	\begin{enumerate}
		\item (Recurrence): for almost every initial condition $(x^0,y^0)$, there exists an increasing sequence of iterations $t_n$ such that $(x^{t_n},y^{t_n})\to (x^0,y^0)$. 	
		\item (Bounded Orbits): agent strategies $\{{x}^t,{y}^t\}_{t=0}^\infty$ are bounded. Specifically,  
		\begin{align*}
		\left\lVert {x}^t\right\rVert^2_{\bar{P}^{-1}D_{\bar{\eta}}^{-1}} + \left\lVert {y}^t\right\rVert^2_{\bar{Q}^{-1}D_{\bar{\gamma}}^{-1}}
		&\leq \frac{\left\lVert {x}^0\right\rVert^2_{\bar{P}^{-1}D_{\bar{\eta}}^{-1}} + \left\lVert {y}^0\right\rVert^2_{\bar{Q}^{-1}D_{\bar{\gamma}}^{-1}}+\left\langle {x}^0, \bar{A}{y}^0\right\rangle}
		{1-\frac{\left\lVert \bar{A}\right\rVert\cdot\left\lVert \bar{P}^{\frac{1}{2}}D_{\bar{\eta}}^{\frac{1}{2}}\right\rVert\cdot\left\lVert \bar{Q}^{\frac{1}{2}}D_{\bar{\gamma}}^{\frac{1}{2}}\right\rVert}{2}}.
		\end{align*}
		\item (Convergence): each agent has $O\left(1/T\right)$ time-average convergence to the set of Nash equilibria.
	\end{enumerate}	
\end{theorem}

\section{Games with Additional Linear Payouts and Games Using Probability Vectors}\label{sec:bilinear}

We briefly remark that our results extend to the setting
\begin{align}
\max_{x_i\in{\cal X}_i} \left\langle x_i, \sum_{j\neq i} A^{(ij)}x_j-b_i\right\rangle \ for \ all \ i. \tag{Network Game with Additional Linear Payouts} \label{eqn:MultiAgentGame2}
\end{align}

A Nash equilibrium $x^*$ of this game satisfies $\sum_{j\neq i} A^{(ij)} x_j^*=b_i$ for each agent $i$. 
We can reduce this game to (\ref{eqn:MultiAgentGame}) simply by expressing $x_i$ relative to $x_i^*$ for each agent $i$. 
Formally, let $\hat{\cal X}_i=\bigcup_{x_i\in {\cal X}_i} \{x_i-x_i^*\}$ (in our setting, ${\cal X}_i$ is affine and therefore ${\cal X}_i=\hat{\cal X}_i$). 
Thus, (\ref{eqn:MultiAgentGame2}) is equivalent to
\begin{align*}
&\max_{x_i\in\hat{\cal X}_i} \left\langle (x_i+x_i^*), \sum_{j\neq i} A^{(ij)}(x_j+x_j^*)-b_i\right\rangle \ for \ all \ i  \label{eqn:MultiAgentGame3}\\
=&\max_{x_i\in\hat{\cal X}_i} \left\langle (x_i+x_i^*), \sum_{j\neq i} A^{(ij)}x_j\right\rangle \ for \ all \ i 
\end{align*} 
From agent $i$'s perspective, $\left\langle x_i^*, \sum_{j\neq i} A^{(ij)}x_j\right\rangle$ is constant. 
Thus, all maximizers of the previous expression also maximize
\begin{align*}
&\max_{x_i\in\hat{\cal X}_i} \left\langle x_i, \sum_{j\neq i} A^{(ij)}x_j\right\rangle \ for \ all \ i.
\end{align*} 
Thus, a game with additional linear payouts can always be expressed as a game without additional linear payouts after shifting the strategy space. 
Moreover, the constant $\left\langle x_i^*, \sum_{j\neq i} A^{(ij)}x_j\right\rangle$ plays no role for online optimization methods that rely on gradients of the utility function, e.g., gradient descent.
Thus, the behavior of gradient descent will remain unchanged and all previous results extend to (\ref{eqn:MultiAgentGame2}). 

This reduction gives some ideas on how to extend these results when ${\cal X}_i$ is the set of probability vectors, i.e., ${\cal X}_i=\{x\in \mathbb{R}^{S_i}_{\geq 0}: \sum_{s_i=1}^{S_i} x_{is_i}=1 \}$. 
After performing the substitution $x_{iS_i}=1-\sum_{s_i=1}^{S_i-1}x_{is_i}$ for each agent $i$, a network game using probability vectors reduces to (\ref{eqn:MultiAgentGame2}) where ${\cal X}_i$ is a compact, full-dimensional space. 
As long as the strategies remain in the interior when using Algorithm \ref{alg:Verlet}, the optimality guarantees will extend as well. 
Regrettably, the space ${\cal X}_i$ is not affine and the energy function may change when the strategies intersect with boundary and more theory needs to be developed to understand this setting.

\section{Experiments: Performance Relative to Optimistic Variants}

In practice, optimistic variants of follow-the-regularized-leader algorithms, e.g., optimistic gradient descent (Algorithm \ref{alg:OptGrdad} below), are often used due to their $O(1/T)$ time-average convergence to the set of Nash equilibria in zero-sum games.  
With the results of Sections \ref{sec:ZeroSum} and \ref{sec:Multi}, Algorithm \ref{alg:Verlet} provides another option for fast convergence.

To obtain this guarantee,  Algorithm \ref{alg:OptGrdad} requires the learning rate $\eta \leq 1/(2||A||)$ \cite{mokhtari2020convergence} while our approach, Algorithm \ref{alg:Verlet}, only requires $\eta \leq 2/||A||$.  
By Theorems \ref{thm:large1} and \ref{thm:large2}, larger learning rates lead to stronger optimization guarantees. 
As such, we hypothesize that by using larger learning rates, Algorithm \ref{alg:Verlet} can outperform \ref{alg:OptGrdad}.

In this section, we perform experiments to support this hypothesis and find that with 97.5\% confidence, Algorithm \ref{alg:Verlet}, on average, results in time-averaged strategies that are 2.585 times closer to the set of Nash equilibria than Algorithm \ref{alg:OptGrdad}.
We also compare Algorithm \ref{alg:Verlet} to an optimized version of Algorithm \ref{alg:OptGrdad} that uses additional memory to avoid matrix products. 
With 97.5\% confidence, Algorithm \ref{alg:Verlet}, on average, results in time-averaged strategies that are 1.742 times faster to the set of Nash equilibria than the optimized version of Algorithm \ref{alg:OptGrdad}.

\begin{varalgorithm}{OptGD}
	\caption{Multiagent Optimistic Gradient Descent.}\label{alg:OptGrdad}
	\begin{algorithmic}[1]
		\Procedure{SimGD}{$A,x^0,\bar{\eta}$}\Comment{Payoff Matrices, Initial Strategies and Learning Rates}
		\For{\texttt{$t=1,...,T$}}
			\For{\texttt{$i=1,...,N$}}
				\State $x_i^t:= x_i^{t-1} + 2\cdot {\bar{\eta}_i} \sum_{j\neq i} A^{(ij)}x_j^{t-1}-\bar{\eta}_i \sum_{j\neq i} A^{(ij)}x_j^{t-2}$ \Comment{Update Strategies Based on Previous Two Iterations}
			\EndFor
		\EndFor
		\EndProcedure
	\end{algorithmic}
\end{varalgorithm}

%\begin{algorithm}[H]
%	\SetAlgoLined
%	%	\KwResult{Write here the result }
%	\textbf{Input:} \\
%	\hspace*{1.2em} Agent $i's$ payoff matrix against agent $j$: $A^{(ij)}$ for $i=1,...,n$ and $j\neq i$\\
%	\hspace*{1.2em} Agent $i$ initial strategy: $\frac{1}{2}x^{-1}_i=x_i^0\in {\cal X}_i$ for $i=1,...,n$\\
%	\hspace*{1.2em} Agent $i$ learning rate: ${\bar{\eta}}_i>0$ for $i=1,...,n$\\
%	
%	
%	\For{$t=1,...,T$}{
%		\For{$i=1,...,N$}{
%			$x_i^t:= x_i^{t-1} + 2\cdot {\bar{\eta}_i} \sum_{j\neq i} A^{(ij)}x_j^{t-1}-\bar{\eta}_i \sum_{j\neq i} A^{(ij)}x_j^{t-2}$ \tcp*{Agent Updates Strategies Based on the Outcome of the Previous Two Iterations}
%		}
%	}
%	\caption{Optimistic Gradient Descent.}\label{alg:OptGrdad}
%\end{algorithm}

We remark that our approach, Algorithm \ref{alg:Verlet} has a distinct advantage over Algorithm \ref{alg:OptGrdad} as we guarantee time-average convergence in a generalization of zero-sum games -- a result not known for Algorithm \ref{alg:OptGrdad}.
However, we conjecture that many results currently in the literature extend to \ref{eqn:MultiPosNeg}s using the techniques we introduced in Section \ref{sec:ZeroSum}.

%\begin{algorithm}[H]
%	\SetAlgoLined
%	%	\KwResult{Write here the result }
%	\textbf{Input:} \\
%	\hspace*{1.2em} Agent $i's$ payoff matrix against agent $j$: $A^{(ij)}$ for $i=1,...,n$ and $j\neq i$\\
%	\hspace*{1.2em} Agent $i$ initial strategy: $\frac{1}{2}x^{-1}_i=x_i^0\in {\cal X}_i$ for $i=1,...,n$\\
%	\hspace*{1.2em} Agent $i$ learning rate: ${\bar{\eta}}_i>0$ for $i=1,...,n$\\
%	
%	$z_i^{-1}= {\bar{\eta}_i} \sum_{j\neq i} A^{(ij)}x_j^{-1}$
%	
%	\For{$t=1,...,T$}{
%		\For{$i=1,...,N$}{
%			$z_i^{t-1}= {\bar{\eta}_i} \sum_{j\neq i} A^{(ij)}x_j^{t-1}$
%			
%			$x_i^t:= x_i^{t-1} + 2\cdot z_i^{t-1}-\bar{\eta}_i \cdot z_i^{t-2}$ \tcp*{Agent Updates Strategies Based on the Outcome of the Previous Two Iterations}
%			
%		}
%	}
%	\caption{Optimistic Gradient Descent with Fewer Matrix Multiplications.}\label{alg:OptGrdad2}
%\end{algorithm}

\subsection{Description of Experiments}

We compare the performance of alternating and optimistic gradient descent with $N\in \{5,10,20\}$ agents where each agent has the same number of strategies ($k\in \{5,10,20\}$). 
We compare the performance of each algorithm across 30 games where $A^{(ij)}_{si,sj}$ is selected uniformly at random from $(-1,1)^{k\times k}$ for $i<j$ and where $A^{(ij)}=[-A^{(ji)}]^\intercal$ for $i>j$ (a zero-sum game) and perform statistical analysis after pairing the samples for each game in order to reduce the variance in the statistical estimates. 
For both algorithms, we select the learning rate to be as large as possible while still ensuring time-average convergence guarantees for any randomly selected set of payoff matrices. 
Specifically, we use the learning rate $\eta=2/(k\cdot (N-1))$ for alternating gradient descent and $\bar{\eta}=1/(2k\cdot (N-1))$ for optimistic gradient descent.  
As discussed in Section \ref{sec:selection}, this selection normalizes the learning rates of the two algorithms.

In our experiments, we also use an optimized version of optimistic gradient descent that uses more memory in exchange for computing fewer matrix products (Algorithm \ref{alg:OptGrdad2} below) and compare our method to both the standard and optimized implementations of optimistic gradient descent. 
Specifically, for a single game and initial condition, we run each of the three algorithms for 30 seconds and measure the distance to the Nash equilibrium with respect to the dual space -- we measure $ || \bar{A}y^t||$ where $\bar{A}$ is the combined payoff matrix introduced in Section \ref{sec:Reduce}. 
As discussed in Section \ref{sec:converge},  $||\bar{A}{y}||=\vec{0}$ if only if $y$ is a Nash equilibrium and $||\bar{A}y||$ measures the distance to the Nash equilibrium in a dual space.

\begin{varalgorithm}{\textoverline{Opt}GD}
	\caption{Multiagent Optimistic Gradient Descent with Fewer Matrix Multiplications.}\label{alg:OptGrdad2}
	\begin{algorithmic}[1]
		\Procedure{\textoverline{Opt}GD}{$A,x^0,\bar{\eta}$}\Comment{Payoff Matrices, Initial Strategies and Learning Rates}
		\For{\texttt{$t=1,...,T$}}	
		\State $z_i^{-1}= {\bar{\eta}_i} \sum_{j\neq i} A^{(ij)}x_j^{-1}$ \Comment{Store $\sum_{j\neq i} A^{(ij)}x_j$}
		\EndFor	
		\For{\texttt{$t=1,...,T$}}
		\For{\texttt{$i=1,...,N$}}
		\State $z_i^{t-1}= {\bar{\eta}_i} \sum_{j\neq i} A^{(ij)}x_j^{t-1}$	\Comment{Store $\sum_{j\neq i} A^{(ij)}x_j$}
		\State $x_i^t:= x_i^{t-1} + 2\cdot z_i^{t-1}-\bar{\eta}_i \cdot z_i^{t-2}$ \Comment{Update Strategies Based on Previous Two Iterations}
		\EndFor
		\EndFor
		\EndProcedure
	\end{algorithmic}
\end{varalgorithm}

Denote $D^{Opt}, D^{\overline{Opt}}$, and $D^{Alt}$ as the distance $||\bar{A}y||$ after 30 seconds of running Algorithms \ref{alg:OptGrdad}, \ref{alg:OptGrdad2}, and \ref{alg:Verlet} respectively. 
Since each instance of $D^{Opt}, D^{\overline{Opt}}$, and $D^{Alt}$ are generated from the same game and initial condition and are also run in sequence, we can pair the results of the individual instances to get an estimate on relative performances $D^{Opt}/D^{Alt}$ and $D^{\overline{Opt}}/D^{Alt}$.

All experiments were conducted in version 4.02 of the R-statistical software on Windows 10 using an i7-10700 processor (2.9GHz) with 32GB of RAM.  
To control for variability caused by computer processing, we generate a single game and run all three algorithms on the game prior to generating the next game. 
The source code and spreadsheet of results for the experiments can be downloaded at \href{http://www.jamespbailey.com/1OverTConvergence}{www.jamespbailey.com/1OverTConvergence}. 
%Since computer performance can decrease over time, we first generate a random game, run the standard implementation of optimistic gradient descent for 30 seconds, then the optimized implementation, and finally our approach (alternating gradient descent).
%While computer performance may introduce a small bias to our data, our ordering of the algorithms ensure that our conclusions remain valid since our approach performs last.

\subsection{Selection of Learning Rates}\label{sec:selection}
In our experiments, we use a single scalar learning rate for all agents.  
To guarantee optimistic gradient descent has $O(1/T)$ time-average convergence to the set of Nash equilibria, $\bar{\eta}$ is required to be at most $1/(2\cdot ||A||)$ \cite{mokhtari2020convergence}. 
However, as shown in Theorem \ref{thm:converge}, alternating gradient descent only requires $\eta < 2/||A||$. 
As such, in our experiments, we always select the learning rate for alternating gradient descent to be four times larger than the learning rate for optimistic gradient descent, i.e., $\eta=4\cdot \bar{\eta}$. 
As shown in Theorems \ref{thm:large1} and \ref{thm:large2}, larger learning suggest better performance for alternating gradient descent and if $\bar{\eta}$ is a valid learning rate for optimistic gradient descent, then $\eta=4\bar{\eta}$ is a valid learning rate for alternating gradient descent. 

We remark that this selection normalizes the values of the learning rates;  
we are forcing both algorithms to operate near the boundary for optimal performance, i.e., $\eta\approx 2/||A||$ and $\bar{\eta}=1/(2||A||)$.
As shown in Lemma \ref{lem:Bounds} in the Appendix, $||A||\leq k\cdot(N-1)$ and therefore we select learning rates $\eta= 2/(k\cdot (N-1))$ and $\bar{\eta}= 1/ (2k\cdot (N-1))$ for alternating and optimistic gradient descent respectively. 
As suggested by Propositions \ref{prop:unbounded} and \ref{prop:noconverge}, $\eta= 2/(k\cdot (N-1))$ will not perform well when $||A||=k\cdot (N-1)$.  
However, the probability that $||A||=k\cdot(N-1)$ is $0$ since the elements of $A$ are generated uniformly at random.

\subsection{Results of Experiments}
In all 270 generated instances, our method (alternating gradient descent) outperformed both implementations of optimistic gradient descent.
Specifically, alternating gradient resulted in strategies that were approximately 2.628 and 1.772 times closer to the set of Nash equilibria than the standard and optimized implementation of optimistic gradient descent respectively. Moreover, across all selections of agents and strategies, we are 97.5\% confident that, on average, alternating gradient descent  will result in strategies that are 2.585 and 1.743 times closer to the set of Nash equilibrium after 30 seconds than the standard and optimized implementation of optimistic gradient descent respectively.
Thus, alternating gradient descent performs significantly better than both implementations of optimistic gradient descent. 
%Results for each pairing of $N$ and $k$ for each game can be found in Appendix \ref{sec:TABLES}.

The relative performance of alternating gradient descent for $N\in \{5,10,20\}$ agents and $k\in \{5,10,20\}$ strategies can be viewed in Tables \ref{tab:OptAlt} and \ref{tab:BarOptAlt}.
For example, in 20 agent, 20 strategy games, we are 97.5\% confident that alternating gradient descent, on average, alternating gradient descent will result in strategies that are 2.5902 and 1.8040 times closer to the set of Nash equilibrium than the standard and optimized implementation of optimistic gradient descent respectively.

\begin{table}[ht]\centering\caption{95\% Confidence Intervals for the Mean of $D^{Opt}/D^{Alt}$ shows that Algorithm \ref{alg:Verlet} significantly outperforms \ref{alg:OptGrdad}.}\label{tab:OptAlt}\vspace{.1in}							
	\begin{tabular}{| c | c c c |}						
		\hline		&\multicolumn{3}{c|}{Strategies}			\\
		\hline	Agents	&5	&10	&20	\\
		\hline	5	&(2.3516,2.6571)	&(2.4694,2.8529)	&(2.5329,2.8443)	\\
		10	&(2.3810,2.6914)	&(2.5003,2.7186)	&(2.6883,2.8618)	\\
		20	&(2.3647,2.5233)	&(2.6771,2.8707)	&(2.5902,2.7336)	\\
		\hline					
	\end{tabular}						
\end{table}

\begin{table}[ht]\centering\caption{95\% Confidence Interval for the Mean of $D^{\overline{Opt}}/D^{Alt}$ shows that Algorithm \ref{alg:Verlet} significantly outperforms \ref{alg:OptGrdad2}.}\label{tab:BarOptAlt}\vspace{.1in}							
	\begin{tabular}{| c | c c c |}						
		\hline		&\multicolumn{3}{c|}{Strategies}			\\
		\hline	Agents	&5	&10	&20	\\
		\hline	5	&(1.7155,1.9284)	&(1.6081,1.8794)	&(1.6495,1.8617)	\\
		10	&(1.6338,1.8602)	&(1.6482,1.7849)	&(1.7557,1.8948)	\\
		20	&(1.6052,1.7333)	&(1.7509,1.8809)	&(1.8040,1.9007)	\\
		\hline					
	\end{tabular}						
\end{table}

\subsection{Importance of Large Learning Rates}

Finally, we test the importance of using larger learning rates;
a key feature of alternating gradient descent is that it enables learning rates four times larger than optimistic gradient descent. 
As suggested by Theorems \ref{thm:large1} and \ref{thm:large2} and shown in Table \ref{tab:Size}, larger learning rates are vital for Algorithm \ref{alg:Verlet}'s superior performance.

\begin{table}[ht]\centering\caption{Impact of Different Learning Rates on the Performance of Alternating Gradient Descent Relative to Optimistic Gradient Descent}\label{tab:Size} \vspace{.1in}							
	\begin{tabular}{| c | c c c |}						
		\hline	&	$\eta=\bar{\eta}$&	$\eta=2\bar{\eta}$&	$\eta=4\bar{\eta}$\\	
		\hline	$D^{Opt}/D^{Alt}$&	(0.6376,0.6934)&	(1.2168,1.3325)&	(2.5003,2.7186)\\	
		$D^{\overline{Opt}}/D^{Alt}$&	(0.4236,0.4706)&	(0.7906,0.8725)&	(1.6482,1.7849)\\	
		\hline					
	\end{tabular}						
\end{table}

\section{Conclusion}
In this paper, we have proven that alternating gradient descent achieves $O(1/T)$ time-average convergence to the set of Nash equilibria in a generalization of network zero-sum games. 
Further, we have experimentally shown with 97.5\% confidence that, on average, alternating gradient results in time-averaged strategies that are 2.585 times closer to the set of Nash equilibria than optimistic gradient descent. 
In addition to providing a faster algorithm for a more general set of games, this paper also demonstrates the potential power of carefully constructing close approximations of continuous-time learning dynamics. 

\bibliographystyle{plain}  
\bibliography{References}

\begin{thebibliography}{10}

\bibitem{abernethy2021last}
Jacob Abernethy, Kevin~A Lai, and Andre Wibisono.
\newblock Last-iterate convergence rates for min-max optimization: Convergence
  of hamiltonian gradient descent and consensus optimization.
\newblock In {\em Algorithmic Learning Theory}, pages 3--47. PMLR, 2021.

\bibitem{adsul2021fast}
Bharat Adsul, Jugal Garg, Ruta Mehta, Milind Sohoni, and Bernhard Von~Stengel.
\newblock Fast algorithms for rank-1 bimatrix games.
\newblock {\em Operations Research}, 69(2):613--631, 2021.

\bibitem{Bailey2019Regret}
James~P Bailey, Gauthier Gidel, and Georgios Piliouras.
\newblock Finite regret and cycles with fixed step-size via alternating
  gradient descent-ascent.
\newblock In {\em Conference on Learning Theory}, pages 391--407. PMLR, 2020.

\bibitem{Bailey18Divergence}
James~P. Bailey and Georgios Piliouras.
\newblock Multiplicative weights update in zero-sum games.
\newblock In {\em Proceedings of the 2018 ACM Conference on Economics and
  Computation}, EC ’18, page 321–338, New York, NY, USA, 2018. Association
  for Computing Machinery.

\bibitem{Bailey19GDRegret}
James~P. Bailey and Georgios Piliouras.
\newblock Fast and furious learning in zero-sum games: Vanishing regret with
  non-vanishing step sizes.
\newblock In {\em Advances in Neural Information Processing Systems 32}, pages
  12977--12987. Curran Associates, Inc., 2019.

\bibitem{Bailey19Hamiltonian}
James~P. Bailey and Georgios Piliouras.
\newblock Multi-agent learning in network zero-sum games is a hamiltonian
  system.
\newblock In Edith Elkind, Manuela Veloso, Noa Agmon, and Matthew~E. Taylor,
  editors, {\em Proceedings of the 18th International Conference on Autonomous
  Agents and MultiAgent Systems, {AAMAS} '19, Montreal, QC, Canada, May 13-17,
  2019}, pages 233--241. International Foundation for Autonomous Agents and
  Multiagent Systems, 2019.

\bibitem{balcan2012weighted}
Maria-Florina Balcan, GATECH EDU, Florin Constantin, HARVARD EDU, Ruta Mehta,
  and IITB AC.
\newblock The weighted majority algorithm does not converge in nearly zero-sum
  games.
\newblock {\em Rn}, 1:S2, 2012.

\bibitem{barreira}
Luis Barreira.
\newblock Poincare recurrence: old and new.
\newblock In {\em XIVth International Congress on Mathematical Physics. World
  Scientific.}, pages 415--422, 2006.

\bibitem{cai2016zero}
Yang Cai, Ozan Candogan, Constantinos Daskalakis, and Christos Papadimitriou.
\newblock Zero-sum polymatrix games: A generalization of minmax.
\newblock {\em Mathematics of Operations Research}, 41(2):648--655, 2016.

\bibitem{cesa2006prediction}
Nicolo Cesa-Bianchi and G{\'a}bor Lugosi.
\newblock {\em Prediction, learning, and games}.
\newblock Cambridge university press, 2006.

\bibitem{Marco2020Chaos}
Yun~Kuen Cheung and Georgios Piliouras.
\newblock Chaos, extremism and optimism: Volume analysis of learning in games,
  2020.

\bibitem{daskalakis2019last}
C~Daskalakis and Ioannis Panageas.
\newblock Last-iterate convergence: Zero-sum games and constrained min-max
  optimization.
\newblock In {\em 10th Innovations in Theoretical Computer Science (ITCS)
  conference, ITCS 2019}, 2019.

\bibitem{daskalakis2021nearoptimal}
Constantinos Daskalakis, Maxwell Fishelson, and Noah Golowich.
\newblock Near-optimal no-regret learning in general games, 2021.

\bibitem{du2017stochastic}
Simon~S Du, Jianshu Chen, Lihong Li, Lin Xiao, and Dengyong Zhou.
\newblock Stochastic variance reduction methods for policy evaluation.
\newblock In {\em International Conference on Machine Learning}, pages
  1049--1058. PMLR, 2017.

\bibitem{golowich2020last}
Noah Golowich, Sarath Pattathil, Constantinos Daskalakis, and Asuman Ozdaglar.
\newblock Last iterate is slower than averaged iterate in smooth convex-concave
  saddle point problems.
\newblock In {\em Conference on Learning Theory}, pages 1758--1784. PMLR, 2020.

\bibitem{goodfellow2014generative}
Ian~J. Goodfellow, Jean Pouget-Abadie, Mehdi Mirza, Bing Xu, David
  Warde-Farley, Sherjil Ozair, Aaron Courville, and Yoshua Bengio.
\newblock Generative adversarial networks, 2014.

\bibitem{Hairer2006EnergyConserve}
Ernst Hairer.
\newblock {\em Long-time Energy Conservation}, page 162–180.
\newblock London Mathematical Society Lecture Note Series. Cambridge University
  Press, 2006.

\bibitem{hairer2006geometric}
Ernst Hairer, Marlis Hochbruck, Arieh Iserles, and Christian Lubich.
\newblock Geometric numerical integration.
\newblock {\em Oberwolfach Reports}, 3(1):805--882, 2006.

\bibitem{kangarshahi2018let}
Ehsan~Asadi Kangarshahi, Ya-Ping Hsieh, Mehmet~Fatih Sahin, and Volkan Cevher.
\newblock Let’s be honest: An optimal no-regret framework for zero-sum games.
\newblock In {\em International Conference on Machine Learning}, pages
  2488--2496. PMLR, 2018.

\bibitem{kannan2010games}
Ravi Kannan and Thorsten Theobald.
\newblock Games of fixed rank: A hierarchy of bimatrix games.
\newblock {\em Economic Theory}, 42(1):157--173, 2010.

\bibitem{Mertikopoulos2018CyclesAdverserial}
Panayotis Mertikopoulos, Christos Papadimitriou, and Georgios Piliouras.
\newblock Cycles in adversarial regularized learning.
\newblock In {\em Proceedings of the Twenty-Ninth Annual ACM-SIAM Symposium on
  Discrete Algorithms}, SODA ’18, page 2703–2717, USA, 2018. Society for
  Industrial and Applied Mathematics.

\bibitem{mertikopoulos2016learning}
Panayotis Mertikopoulos and William~H Sandholm.
\newblock Learning in games via reinforcement and regularization.
\newblock {\em Mathematics of Operations Research}, 41(4):1297--1324, 2016.

\bibitem{mokhtari2020convergence}
Aryan Mokhtari, Asuman~E Ozdaglar, and Sarath Pattathil.
\newblock Convergence rate of {O(1/k)} for optimistic gradient and
  extragradient methods in smooth convex-concave saddle point problems.
\newblock {\em SIAM Journal on Optimization}, 30(4):3230--3251, 2020.

\bibitem{Poincare1890}
Henri Poincar{\'e}.
\newblock Sur le probl{\`e}me des trois corps et les {\'e}quations de la
  dynamique.
\newblock {\em Acta mathematica}, 13(1):A3--A270, 1890.

\bibitem{Shahrampour20OnlineAllocation}
S.~{Pu}, J.~J. {Escudero-Garzas}, A.~{Garcia}, and S.~{Shahrampour}.
\newblock An online mechanism for resource allocation in networks.
\newblock {\em IEEE Transactions on Control of Network Systems}, pages 1--1,
  2020.

\bibitem{rudin1987real}
Walter Rudin.
\newblock Real and complex analysis (mcgraw-hill international editions:
  Mathematics series).
\newblock 1987.

\bibitem{wei2020linear}
Chen-Yu Wei, Chung-Wei Lee, Mengxiao Zhang, and Haipeng Luo.
\newblock Linear last-iterate convergence in constrained saddle-point
  optimization.
\newblock {\em arXiv preprint arXiv:2006.09517}, 2020.

\end{thebibliography}
\appendix

\section{Proof of Proposition \ref{prop:chaos2}}\label{app:Chaos}

\Chaos*

\begin{proof}
	With the selection of learning rates and the payoff matrix $A$, one iteration  of Algorithm \ref{alg:2Agent} maps the point $(x,y) \to (x+y, x+2y)$. 
	Since the mapping is linear and $V^0$ is convex with a finite number of extreme points, $V^t$ will also be convex with a finite number of extreme points. 
	Moreover, the extreme points of $V^t$ can be determined directly from the extreme points of $V^{t-1}$. 
	Let $E^t$ denote the extreme points of $V^t$.  
	We begin by claiming that $E^t=(\pm (F_{2t-2}, F_{2t-1}), \pm (F_{2t+1}, F_{2t+2}))$ where $F_k$ is the $k$th number in the Fibonacci sequence. 
	The sequence is defined by $F_{k}=F_{k-1}+F_{k-2}$ with $F_1=F_2=1$. 
	
	We proceed induction and consider $t=0$.  
	Extending the sequence backwards, $F_{-2}=-1$ and $F_{-1}=1$. 
	Thus, for $t=0$, $(\pm (F_{2t-2}, F_{2t-1}), \pm (F_{2t+1}, F_{2t+2}))=(\pm (-1, 1), \pm (1, 1))=E^0$ completing the base case. 	
	Next, by the inductive hypothesis, $E^{t-1}=	(\pm (F_{2t-4}, F_{2t-3}), \pm (F_{2t-1}, F_{2t}))$.
	As stated before, we compute $E^t$ directly from $E^{t-1}$ using the map $(x,y)\to (x+y, x+2y)$:
	First,
	\begin{align*}
		(F_{2t-4}, F_{2t-3})\to (F_{2t-4}+F_{2t-3}, F_{2t-4}+F_{2t-3}+F_{2t-3}) = (F_{2t-2}, F_{2t-2}+F_{2t-3})  = (F_{2t-2}, F_{2t-1}).
	\end{align*}
	Following identically, $-(F_{2t-4}, F_{2t-3})\to-(F_{2t-2}, F_{2t-1})$. 
	Next, 
	\begin{align*}
	(F_{2t-1}, F_{2t})\to (F_{2t-1}+F_{2t},F_{2t-1}+ F_{2t}+F_{2t}) = (F_{2t+1},F_{2t+1}+F_{2t})= (F_{2t+1},F_{2t+2}).
	\end{align*}
	Similarly, $-(F_{2t-1}, F_{2t})\to -(F_{2t+1},F_{2t+2})$. 
	Thus, $E^t= (\pm (F_{2t-2}, F_{2t-1}), \pm (F_{2t+1}, F_{2t+2}))$ as claimed. 
	
	The diameter of $V^t$ is then given by the distance between $(F_{2t+1}, F_{2t+2})$ and $-(F_{2t+1}, F_{2t+2})$ since $F_k$ is increasing for $k>0$.
	This distance is $\sqrt{2F_{2t+1}^2+2F_{2t+2}^2}$.
	It is well-known that $F_k\in \Theta(\phi^k)$ where $\phi= (1+\sqrt{5})/2$ is the golden ratio. 
	Thus, 
	\begin{align*}\sqrt{2F_{2t+1}^2+2F_{2t+2}^2}\in \Theta(\sqrt{2\phi^{4t+2}+2\phi^{4t+3}})=\Theta(\phi^{2t})
	\end{align*}
	and the diameter grows exponentially. 
	
	Finally, the volume of $V^0$ is 4.  
	By Theorem \ref{thm:Volume}, Volume is invariant  and therefore the volume of $V^t$ is also 4 thereby completing the proof. 
\end{proof}

\section{Bounding  $||A||$ for Experiments}

\begin{restatable}[]{lemma}{Bounds}\label{lem:Bounds}
	Let $S_i=k$ for all $i=1,...,N$ and $A^{(ij)}_{s_i,s_j} \sim U(-1,1)$ for $i<j$ and let $A^{(ji)}=[-A^{(ij)}]^\intercal$. Let $A$ be the combined payoff matrix as defined in Section \ref{sec:Reduce}.  Then $||A|| \leq k\cdot(N-1)$. Further, there exists an instance $A$ such that $||A|| \geq \frac{k}{\sqrt{3}}\cdot(N-1)$. 
\end{restatable}

The first part of the proof consists of showing that $||A||\leq ||B||$ where $B$ is a coordination game with $B^{(ij)}_{s_i,s_j}=1$.  
We then show that the largest eigenvalue of $B$ is $k\cdot(N-1)$ implying $||A||\leq ||B||=k\cdot(N-1)$. 
The second part of the proof simply consists of giving an $A$ and $x$ such that  $||A||=\max_{v} ||Av||/||v|| \geq ||Ax||/||x||=\frac{k}{\sqrt{3}}\cdot(N-1)$. 

\begin{proof}
	Let $A$ be an arbitrary $U\times V$ matrix.
	We first show that if $A'$ is such that $A'_{uv} \geq |A_{uv}|$ then $||A'|| \geq ||A||$. 
	Let $x$ be such that $||A||=||Ax||$ as in the definition of $||A||$.
	Then,
	\begin{align*}
	||A||=||Ax||&= \sqrt{\sum_{u=1}^U \left( \sum_{v=1}^V A_{uv} x_v\right)^2}\\
	&\leq \sqrt{\sum_{u=1}^U \left( \sum_{v=1}^V |A_{uv}|\cdot  |x_v|\right)^2}\\\\
	&\leq \sqrt{\sum_{u=1}^U \left( \sum_{v=1}^V A'_{uv}\cdot  |x_v|\right)^2}\\
	&= {\left\lVert A' |x| \right\rVert} \leq \max_{y: ||y||=1} ||A'y|| = ||A'||
	\end{align*} 
	
	Next, let  $\mathbf{1}_w$ be a vector of $w$ 1's and let $B^{(ij)}=\mathbf{1}_k\mathbf{1}_k^\intercal$ for all $i\neq j$ (a $k\times k$ matrix of 1's)  and let $B^{{(ii)}}= 0\cdot \mathbf{1}_k\mathbf{1}_k^\intercal$.  
	This corresponds to a coordination game where the payout for every pair of pure strategies is 1. 
	By selection, $B_{uv} \geq |A_{uv}|$ when $A^{(ij)}_{si,s_j}$ is generated between $-1$ and $1$ and therefore, by the previous claim, $||B||\geq ||A||$.  
	
	It is well-known that the spectral norm and Euclidean norms are equivalent, i.e., $|\lambda_{max}|=||B||$ where $\lambda_{max}$ is the largest eigenvalue for $B$.   Let $\lambda$ and $v$ be any eigenvalue/eigenvector pair for $B$. We write the components of the eigenvector $v$ as $v_{js_j}$ in order to freely move between the notations $B$ and $B^{(ij)}$.  
	We also use $v_j=\{v_{j1},v_{j2},...,v_{jk}\}$ to denote the portion of the eigenvector that multiplies by $B^{(ij)}$ in the definition of an eigenpair. 
	Since $\lambda$ and $v$ are an eigenpair,
	\begin{align*}		
	Bv = \lambda v &\Rightarrow \sum_{j\neq i} B^{(ij)} v_j = \lambda v_i \ \forall_{i=1}^N\\
	&\Rightarrow \sum_{j\neq i} \mathbf{1}_k\mathbf{1}_k^\intercal v_j = \lambda v_i \ \forall_{i=1}^N.
	\end{align*}
	The term $\mathbf{1}_k\mathbf{1}_k^\intercal v_j=\mathbf{1}_k(\mathbf{1}_k^\intercal v_j)$ is a vector of constants and there exists a $c_i$ such that $v_i=c_i\cdot \mathbf{1}_k$.
	Continuing from above, this implies
	\begin{align*}
	\sum_{j\neq i} c_j \mathbf{1}_k\mathbf{1}_k^\intercal\mathbf{1}_k= \lambda c_i \mathbf{1}_k \ \forall_{i=1}^N & \Rightarrow \sum_{j\neq i} c_j\cdot k  \mathbf{1}_k= \lambda c_i\mathbf{1}_k \ \forall_{i=1}^N\\
	& \Rightarrow \sum_{j\neq i} c_j\cdot k= \lambda c_i \ \forall_{i=1}^N\\
	& \Rightarrow|\lambda | \cdot | c_i| = |k \sum_{j\neq i} c_j| \leq k \sum_{j\neq i } |c_j|\ \forall_{i=1}^N
	\end{align*} 
	Further, we may select $\lambda, v$ and $c$ so that $\sum_{i=1}^N c_i^2=N$ which implies there exists an $i$ such that $|c_i| \geq 1$ and $\sum_{j\neq i} c_j ^2 \leq N-1$. 
	Selecting such an $i$ implies
	\begin{align*}
	|\lambda| \leq \frac{k \sum_{j\neq i } |c_j|}{|c_i|}& \leq  {k \sum_{j\neq i } |c_j|}  \leq k\cdot (N-1)
	\end{align*}
	where the last inequality follows since $\sum_{j\neq i} x_i : \sum_{j\neq i } x_j^2 \leq N-1, x_i \geq 0$ is a concave function over a convex, symmetric domain and therefore has the symmetric maximizer $x_j= 1$ for all $j\neq i$.
	Thus, $||B||=|\lambda_{max} |\leq k\cdot (N-1)$. 
	We remark that that this bound for $||B||$ is tight -- it is straightforward to verify that  $\lambda = k\cdot(N-1)$ and  $v=\mathbf{1}_{Nk}$ are an eigenpair.
	This completes the first part of the proof since $||A||\leq ||B|| \leq k\cdot (N-1)$.

	To show $||A|| \geq \frac{k}{\sqrt{3}}(N-1)$, let $A^{(ij)}=\mathbf{1}_k\mathbf{1}_k^\intercal$ for $i<j$ and let $A^{(ij)}= [-A^{(ji)}]^\intercal$ for $i>j$.  
	As such, $A$ is block diagonal matrix with 0's on the block diagonal, positive 1's above the diagonal, and negative 1's below the diagonal. 
	Then 
	\begin{align*}
	||A||= \max_{x\in {\cal X}} \frac{||Ax||}{||x||}& \geq \frac{||A\mathbf{1}_{Nk}||}{||\mathbf{1}_{Nk}||}\\
	&= \frac{\sqrt{\sum_{i=1}^n||\sum_{j\neq i} A^{(ij)}\mathbf{1}_k||^2}}{\sqrt{Nk}}\\
	&= \frac{\sqrt{\sum_{i=1}^n||\sum_{j> i} A^{(ij)}\mathbf{1}_k+\sum_{j< i} A^{(ij)}\mathbf{1}_k||^2}}{\sqrt{Nk}}\\
	&= \frac{\sqrt{\sum_{i=1}^n||\sum_{j> i} \mathbf{1}_k\mathbf{1}_k^\intercal\mathbf{1}_k-\sum_{j< i} \mathbf{1}_k\mathbf{1}_k^\intercal\mathbf{1}_k||^2}}{\sqrt{Nk}}\\
	&= \frac{\sqrt{\sum_{i=1}^nk^2||\sum_{j> i} \mathbf{1}_k-\sum_{j< i} \mathbf{1}_k||^2}}{\sqrt{Nk}}\\
	&= \frac{\sqrt{\sum_{i=1}^nk^2(N-2i+1)^2||\mathbf{1}_k||^2}}{\sqrt{Nk}}\\
	&= \frac{\sqrt{k^3N(N^2-1)/3}}{\sqrt{Nk}}\\
	&= \frac{k}{\sqrt{3}}\cdot \sqrt{N^2-1}  \geq \frac{k}{\sqrt{3}}\cdot \sqrt{N^2-2N+1}=\frac{k}{\sqrt{3}}\cdot (N-1)
	\end{align*}
\end{proof}

\newpage

\end{document}